\documentclass[twoside]{article}

\usepackage[preprint]{preprint}

\usepackage[round]{natbib}
\renewcommand{\bibname}{References}

\usepackage{amssymb}
\usepackage{amsthm}
\usepackage{mathtools}
\usepackage{enumitem}
\usepackage[dvipsnames]{xcolor}
\usepackage{algorithm}
\usepackage{algpseudocodex}
\algrenewcommand{\textproc}[1]{\textup{\textsc{#1}}}
\usepackage[bibliography=common]{apxproof}
\usepackage{graphicx}
\usepackage{booktabs}
\usepackage{caption}
\usepackage{subcaption}
\usepackage{tikz}
\usetikzlibrary{positioning, arrows.meta, shapes.geometric, calc, shapes.misc}
\tikzset{
  every picture/.append style={color=black},
  obs/.style = {name = #1, circle, draw, inner sep = 5pt, label = center:\tiny{$#1$}},
  causal/.style = {-{Stealth[length=2mm]}, semithick},
  bidirected/.style = {{Stealth[length=2mm]}-{Stealth[length=2mm]}, dashed, semithick}
}

\newcommand{\An}{\mathrm{An}}
\newcommand{\Ch}{\mathrm{Ch}}
\newcommand{\doo}{\mathrm{do}}
\newcommand{\+}[1]{\ensuremath{\mathbf{#1}}}
\newcommand{\cG}{\mathcal{G}}
\newcommand{\cF}{\mathcal{F}}
\newcommand{\cH}{\mathcal{H}}
\newcommand{\cC}{\mathcal{C}}
\newcommand{\cS}{\mathcal{S}}
\newcommand{\cP}{\mathcal{P}}
\newcommand{\cW}{\mathcal{W}}
\newcommand{\cI}{\mathcal{I}}

\newtheorem{definition}{Definition}

\newtheoremrep{cor}{Corollary}
\newtheoremrep{thm}{Theorem}
\newtheoremrep{lem}{Lemma}

\renewcommand{\appendixprelim}{}
\begin{document}

\runningtitle{What Should We Measure Next?}
\runningauthor{Väänänen, Cui, Cinelli, and Tikka}

\twocolumn[

\aistatstitle{What Should We Measure Next?\\ Finding Identification Strategies by Refining Mechanisms}

\aistatsauthor{Mikko Väänänen\textsuperscript{1} \and Fanyu Cui\textsuperscript{2} \and Carlos Cinelli\textsuperscript{3} \and Santtu Tikka\textsuperscript{1}}

\aistatsaddress{
\textsuperscript{1}Department of Mathematics and Statistics, University of Jyväskylä\\
\textsuperscript{2}Department of Biostatistics, Columbia University\\
\textsuperscript{3}Department of Statistics, University of Washington
} ]

\begin{abstract}
Canonical approaches in causal inference  treat model specification as fixed, assuming that researchers directly translate all relevant domain knowledge into a causal model, which can then be used to deduce its logical implications. Yet, in  practical applications, model specification is often an iterative process, and involves exploration and introspection---of the many aspects of the phenomenon one could investigate, which ones actually matter for the identification of the causal effect of interest?  In this paper we study the problem of iterative identification in partially specified causal models. We focus on determining where observing variables that intercept a direct effect or a confounding path between two variables could enable identification in semi-Markovian models. We give necessary and sufficient graphical conditions for when observing such variables can render an unidentifiable query identifiable, together with an efficient algorithm for locating all such opportunities in a given causal diagram. Our results can help analysts better navigate the model space by drawing attention to the parts of their substantive knowledge that could result in a successful identification strategy.
\end{abstract}

\section{Introduction}

Causal inference combines observations with assumptions about the data generating process. Once these assumptions are represented in a causal diagram, identification theory determines whether the effect of interest can be recovered from the available observational distribution \citep{pearl:b1995,shpitserpearl:2006}. When it cannot, a natural question arises: what additional information would make identification possible?

Answering this question requires deciding which further aspects of the phenomenon to bring into consideration. Researchers may possess considerable knowledge that has not been expressed in their current causal diagram, without recognizing its relevance to the effect of interest \citep{cinelli2025challenges}. A directed edge may summarize a mechanism they could describe in greater detail, and a bidirected edge may represent common causes they could investigate. A priori, a researcher could expand the mechanisms behind any and all of those arrows, but that is a formidable task. Ideally, we want a system that helps direct the researcher's attention to the information that matters.

In this paper, we use causal diagrams to guide that reflection. We characterize where observing mediators or common causes would render a previously nonidentifiable causal effect identifiable. These results direct the researcher toward particular mechanisms and sources of confounding, helping them recognize which parts of their domain knowledge could provide a way forward. To illustrate, consider the simplest setting, in which a researcher starts their analysis with the causal diagram in Figure~\ref{fig:simple_example}(a). Here, $X$ denotes the treatment and $Y$ the outcome, while the bidirected edge $X\leftrightarrow Y$ allows for unmeasured confounding between them. The causal effect of interest is therefore not identifiable without further assumptions. What next steps could an investigator take in such circumstances?

A well-known strategy for identifying the causal effect of $X$ on $Y$ is to block the confounding paths between them, an approach with a long history in causal inference \citep{rosenbaumrubin1983,robins1986,pearl:b1995}. In Figure~\ref{fig:simple_example}(b), these common causes are represented by expanding the edge $X \leftrightarrow  Y$  with a (possibly vector-valued) variable $N$ (for ``new''). The causal effect is now identifiable by backdoor adjustment. However, this is not the only identification strategy available here. An alternative route is to investigate the mediating channels through which the causal effect of $X$ flows to $Y$, as in Figure~\ref{fig:simple_example}(c). This is represented by expanding the directed edge $X \rightarrow Y$ with mediators $N$, which now allows identification of the causal effect through Pearl's celebrated frontdoor formula, even though the original confounding path $X \leftrightarrow Y$ remains intact.

For general causal models, we consider extensions of the original diagram that replace a directed edge with a path through mediators or a bidirected edge with observed common causes. We focus on expanding one edge at a time, as this represents the minimal refinement of the original causal diagram that may suffice to identify the effect of interest. We call these operations {\emph{edge expansions}} and call an edge \emph{ID-enabling} when its expansion makes the target effect identifiable.  These edges indicate where further substantive knowledge could yield a successful identification strategy.

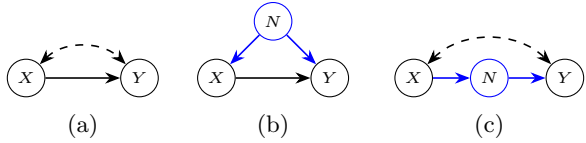
\begin{figure}[t]
\centering
\begin{subfigure}[b]{0.14\textwidth}
  \centering
  \begin{tikzpicture}
    \node[obs = {X}] at (0, 0) {};
    \node[obs = {Y}] at (1.5, 0) {};
    \draw[causal] (X) -- (Y);
    \draw[bidirected, bend left=45] (X) to (Y);
  \end{tikzpicture}
\caption{}
\end{subfigure}
\begin{subfigure}[b]{0.14\textwidth}
  \centering
  \begin{tikzpicture}
    \node[obs = {X}] at (0, 0){};
    \node[obs = {Y}] at (1.5, 0) {};
    \node[obs = {N}, draw=blue, every label/.append style={text=black}] at (0.75, 0.75) {};
    \draw[causal] (X) -- (Y);
    \draw[causal, blue] (N) -- (X);
    \draw[causal, blue] (N) -- (Y);
  \end{tikzpicture}
\caption{}
\end{subfigure}
\begin{subfigure}[b]{0.18\textwidth}
  \centering
  \begin{tikzpicture}
    \node[obs = {X}] at (0,0) {};
    \node[obs = {N}, draw=blue, every label/.append style={text=black}] at (1,0) {};
    \node[obs = {Y}] at (2,0) {};
    \draw[causal, blue] (X) -- (N);
    \draw[causal, blue] (N) -- (Y);
    \draw[bidirected, bend left=45] (X) to (Y);
  \end{tikzpicture}
\caption{}
\end{subfigure}
\caption{Two ways to make $p(y|\doo(x))$ identifiable. (a) The original graph, in which the effect is not identifiable. (b) Observed common causes $N$ allow identification by backdoor adjustment. (c) Observed mediators $N$ allow identification through the frontdoor criterion. The newly observed node and its related edges are highlighted in blue.}
\label{fig:simple_example}
\end{figure}

Our contributions are as follows:
\begin{enumerate}
\item We derive simple graphical conditions for quickly finding ID-enabling directed and bidirected edges, and ruling out edges whose expansion cannot enable identification. These conditions are fast and powerful: across our experimental settings, they fully resolve $87.7$--$100\%$ of nonidentified queries and find an ID-enabling edge in more than $99.9\%$ of cases whenever a solution exists.
\item Next, we derive necessary and sufficient conditions for ID-enabling edges based on graphical obstructions to identification, known as \emph{hedges}. We then develop a sound, complete, and efficient algorithm for finding all such edges. For example, for a graph with $n$ nodes and $m$ edges, our algorithm retains at most $\min\{m,2n-2\}$ candidates, compared with $m$ when testing every edge.  We prove that the $2n-2$ bound is in fact sharp.
\item Finally, we perform exhaustive experiments on graphs with two through six nodes, quantifying how often identification can be enabled by expanding a directed or bidirected edge, as well as the coverage of our simple graphical conditions. In runtime comparisons on random samples of large graphs, we observe that our algorithm provides more than a thousand-fold speedup over testing every edge in,  for instance, dense graphs with 160 nodes.
\end{enumerate}
Discussion on related literature is available in Appendix~\ref{app:literature}.

\begin{toappendix}
\section{Related literature}
\label{app:literature}

Pearl's original account of causal diagrams explicitly suggests their use in seeking additional observations or experiments when an effect is not identifiable \citep{pearl:b1995}.

Within a given diagram, causal methods distinguish useful from harmful controls \citep{cinelli_forney_pearl_2024}, characterize and enumerate valid adjustment sets \citep{zander2019}, and find sets satisfying the front-door criterion \citep{jeong2022,wienobst2024}. These methods select among variables whose causal roles are already represented in the diagram. Measurement costs and statistical efficiency can further guide the choice of an adjustment set \citep{smucler2022,henckel:apa2019}. \citet{karvanen2021} extend this perspective to study design, using do-search to evaluate whether proposed additional data sources would identify a causal effect under specified causal assumptions---note, however, that here the proposals are made by the researcher, instead of found and suggested by an algorithm.

\citet{guo2026} are somewhat related to our emphasis on iterative model specification. Their interactive procedure asks users for primary adjustment sets and uses these inputs to construct sufficient adjustment sets, without requiring a prespecified graph. One of the differences here is that they start with a specific strategy in mind---identification via covariate adjustment---and then proceed to try to elicit knowledge for that strategy from the researcher. Our approach, instead, starts from a partially specified causal model and then directs the researcher's attention to where additional refinement enables identification. The question we solve and the problem scope are therefore different. For the edge expansions we consider, identification can proceed through any strategy available in non-parametric SCMs, and thus goes beyond covariate adjustment.

A further connection is the use of graphical obstructions to identification to guide changes in a study or its assumptions. Hedges have been used to formulate the problem of choosing interventions for identification and to develop algorithms that minimize their cost \citep{akbari2022,elahi2024,akbari2025}. \citet{akbari2023} study minimum cost deletions of directed and bidirected edges when the causal structure is uncertain, using hedges to guide the search. These approaches seek interventions on existing variables or assumptions that remove edges. Here we use hedges to find all individual edge expansions that enable identification through additional observations, refining the mechanisms represented in the original diagram.

Additional information can also change what is learned without measuring a new variable. Functional form restrictions, knowledge about plausible magnitudes, and equality constraints can aid identification or yield informative bounds on causal effects \citep{vanderweele2009monotonicity,cinelli2019scm,cinelli:jrssb2020,cinelliHazlett2025iv,kumor2020acid,zhang2021equality,maiti2025monotonicity,chernozhukov2026ovb,wang2026did}. Extending identification searches of the kind developed here to these other forms of knowledge is an important direction for future research.

\end{toappendix}

\section{Preliminaries}

Throughout the paper, boldface notation denotes sets, whereas non-bold notation denotes individual elements. Capital letters denote variables and small letters denote their values, value assignments, or edges of graphs.

A semi-Markovian structural causal model \citep[SCM, ][]{shpitserpearl:2006, pearl:2009} consists of observed variables $\+ V$, latent exogenous variables $\+ U$, and structural equations assigning each $V_i \in \+ V$ as a function of its observed parents and a subset of $\+ U$, where exogenous variables are mutually independent and each may be a parent of at most two observed variables. The causal structure of a semi-Markovian SCM is represented by an acyclic directed mixed graph \citep[ADMG,][]{richardson2003admg, richardson2023admg} $\cG = (\+ V, \+ D, \+ B)$ where $\+ V$ is the set of nodes, $\+ D$ is the set of directed edges and $\+ B$ is the set of bidirected edges. Directed edges represent possible direct causal effects between members of $\+ V$ and bidirected edges represent dependence induced by an unobserved common cause in $\+ U$. For disjoint $\+ X, \+ Y \subset \+ V$, the causal effect of $\+ X$ on $\+ Y$ is described by the interventional distribution $p(\+y|\doo(\+ x))$, where the do-operator denotes replacing the structural equations for $\+ X$ by the fixed assignments $\+ X = \+ x$, while leaving all other equations unchanged. Graphically, the intervention removes all edges entering $\+ X$.

For any ADMG $\cG$, $V(\cG) \equiv \+ V$, $E(\cG) \equiv \+ D \cup \+ B$, $\cG_{\bar{\+ X}, \underline{\+ Z}}$ is the graph obtained from $\cG$ by removing all incoming edges to $\+ X$ and all outgoing edges from $\+ Z$. The set of ancestors of nodes in $\+ Y$ in $\cG$ is denoted by  $\An(\+ Y)_\cG$ including $\+ Y$ by convention. Similarly, we use $\Ch(\+ Y)_\cG$ to denote the children of $\+ Y$ in $\cG$.
In addition, we denote by $\cG_{-e}$ the graph obtained from $\cG$ by removing the edge $e$. If $e \notin E(\cG)$, then $\cG_{-e} = \cG$. A subgraph $\cG_1$ of $\cG_2$ is denoted by $\cG_1 \subseteq \cG_2$. An edge subgraph is a graph obtained from $\cG_2$ by removing zero or more edges while preserving all nodes. The induced subgraph of a set of nodes $\+ W$ is denoted by $\cG[\+ W]$. For a path $\pi$, we use $E(\pi)$ to denote its set of edges.

Next, we define the extended graph obtained by expanding an edge. Intuitively, the new node represents variables involved in the causal or confounding relation summarized by the original edge.

\begin{definition}[Extended Graph] \label{def:extended_graph}
Let $\cG=(\+ V, \+ D,\+ B)$ and let $N \notin \+ V$ be a new node. Expanding an edge $e$ of $\cG$ produces an extended graph $\cG_e$ as follows:
\setlist{nolistsep}
\begin{itemize}[noitemsep]
  \item If $e=(A,V)\in \+ D$, then $\cG_e=(\+ V', \+ D', \+ B)$, where $\+ V' = \+ V \cup \{N\}$ and $\+ D' = (\+ D \setminus \{e\}) \cup \{(A, N), (N, V)\}$.
  \item If $e = \{A,V\} \in \+ B$, then $\cG_e=(\+ V', \+ D', \+ B')$, where $\+ V' = \+ V \cup \{N\}$, $\+ D' = \+ D \cup \{(N, A),(N, V)\}$ and $\+ B' = \+ B \setminus \{e\}$.
\end{itemize}
\end{definition}
In both extensions, the new node $N$ has no incident bidirected edges.

These extensions introduce additional causal detail without imposing new restrictions on the distribution of the original variables. Indeed, $\cG$ is the latent projection of $\cG_e$ onto $\+ V$. The extension therefore preserves exactly the restrictions imposed by the graph on the original observational distribution $p(\+ v)$ \citep[Theorem~4.13]{evans2016graphs}.

We consider identifiability in terms of the assumed causal structure and the observed data distribution.
\begin{definition}[Identifiability]
  Let $\cG = (\+ V, \+ D, \+ B)$ be a causal diagram and let $\+ X, \+ Y \subset \+ V, \+ X \cap \+ Y = \emptyset$. The causal effect $p(\+ y|\doo(\+ x))$ is \emph{identifiable} from $p(\+ v)>0$ in $\cG$ if $p_1(\+ y|\doo(\+ x)) = p_2(\+ y|\doo(\+x))$ for every pair of causal distributions $p_1,p_2$ induced by SCMs compatible with $\cG$ such that $p_1(\+ v) = p_2(\+ v) > 0$ for all value assignments $\+ v$ to $\+ V$.
\end{definition}
In an extended graph, $N$ denotes the new node, and identification concerns the joint distribution ${p(\+ v, n)>0}$ in such graphs.
\section{Graphical Criteria for Finding ID-enabling Edges}
\label{sec:graphical-criteria}
We present some easy-to-use graphical conditions for determining whether an edge expansion can identify a previously nonidentifiable query. The identifiability or nonidentifiability of a causal effect can be quickly determined in the presence of specific graphical patterns \citep{tianpearl2003, pearl:2009}. These patterns can be used to construct sufficient conditions for finding ID-enabling edges and also to determine when an  edge cannot be ID-enabling. The first such condition is related to bidirected paths between the intervention variable and its children.

\begin{thmrep} \label{thm:bidirected_paths}

  If a bidirected edge $b \in \+ B$ belongs to every bidirected path between $X$ and any of its children in $\cG[\An(\+ Y)_\cG]$, then $p(\+ y|\doo(x))$ is identifiable from $p(\+ v, n)$ in the extended graph $\cG_b$.
\end{thmrep}
\begin{proof}
  There are no bidirected paths between $X$ and any of its children in $\cG_b[\An(\+ Y)_{\cG_b}]$ because $b$ intercepts all such paths in $\cG[\An(\+ Y)_\cG]$ and expanding $b$ cannot create new bidirected paths between $X$ and its children. Thus, $p(\+ y|\doo(x))$ is identifiable from $p(\+ v, n)$ in $\cG_b$ by Theorem~4 of \citet{tian2002}.
\end{proof}

In essence, Theorem~\ref{thm:bidirected_paths} states that replacing a bidirected edge with its common causes will make the causal effect of interest identifiable, if the edge intercepts all bidirected paths from $X$ to its children. For directed edges, it suffices that only one child of $X$ is connected to $X$ by a bidirected path in the ancestor graph.

\begin{thmrep}\label{thm:directed_child}
  Suppose that only one child $C$ of $X$ in $\cG[\An(\+ Y)_\cG]$ is connected to $X$ via a bidirected path in that graph. Let $d = (X, C)$. Then $p(\+ y|\doo(x))$ is identifiable from $p(\+ v, n)$ in the extended graph $\cG_d$.
\end{thmrep}
\begin{proof}
Let $N$ denote the new node in $\cG_d$. Subdividing $X\rightarrow C$ preserves directed ancestry between the original nodes. Since $C\in\An(\+ Y)_\cG$, we have
\[
\An(\+ Y)_{\cG_d}=\An(\+ Y)_\cG\cup\{N\}.
\]
In the induced ancestor graph, the children of $X$ in $\cG_d$ are $N$ and its original children other than $C$. The variable $N$ has no incident bidirected edges. The bidirected paths between original nodes are unchanged, so none of the remaining children is connected to $X$ by a bidirected path, by hypothesis. Thus no child of $X$ is connected to $X$ by a bidirected path in $\cG_d[\An(\+ Y)_{\cG_d}]$. Theorem~4 of \citet{tian2002} establishes identification from $p(\+ v, n)$.
\end{proof}

Interestingly, it is also easy to show that whenever a directed edge expansion enables identification, it can never result in a backdoor formula.

\begin{thmrep}\label{thm:directed-no-backdoor}
  If $p(\+ y|\doo(\+ x))$ is not identifiable in $\cG = (\+ V, \+ D, \+ B)$, then no set satisfies the backdoor criterion relative to $\+ X$ and $\+ Y$ in any directed edge expansion $\cG_d$, $d\in \+ D$.
\end{thmrep}
\begin{proof}
Suppose $\+ Z$ satisfies the backdoor criterion in $\cG_d$. The case $N\notin\+ Z$ is immediate. Otherwise, let $A$ be the parent of $N$, and define
\[
\+ Z'=(\+ Z\setminus\{N\})\cup\{A\}.
\]
This set contains no descendants of $\+ X$. Moreover,
\[
\bigl(\+ Y\setminus\+ Z'\perp\!\!\!\perp\+ X\mid\+ Z'\bigr)_{\cG_{\underline{\+ X}}}.
\]
Hence $p(\+ y|\doo(\+ x))$ is identifiable in $\cG$, contradicting the hypothesis.
\end{proof}

Theorems~\ref{thm:bidirected_paths} and~\ref{thm:directed_child} describe cases in which these extensions enable identification. Next, we give a necessary condition for an edge to be ID-enabling.

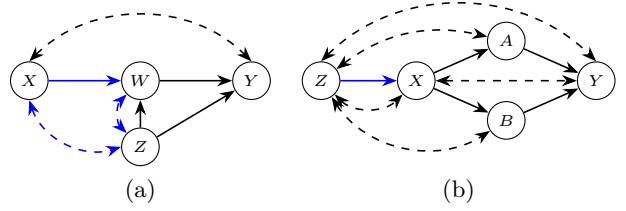
\begin{figure}[t]
\centering
\begin{subfigure}[t]{0.45\columnwidth}
\centering
\begin{tikzpicture}[x=.7cm,y=.7cm]
\path[use as bounding box] (-.3,-1.55) rectangle (4.5,2.3);
\node[obs={X}] at (0,0) {};
\node[obs={W}] at (2.1,0) {};
\node[obs={Y}] at (4.2,0) {};
\node[obs={Z}] at (2.1,-1.25) {};
\draw[causal,blue] (X) -- (W);
\draw[causal] (W) -- (Y);
\draw[causal] (Z) -- (W);
\draw[causal] (Z) -- (Y);
\draw[bidirected, bend left=45] (X.north) to (Y.north);
\draw[bidirected,blue, bend right=45] (X.south) to (Z.west);
\draw[bidirected,blue, bend left=45] (Z.north west) to (W.south west);
\end{tikzpicture}
\caption{}
\label{fig:running-example-direct}
\end{subfigure}\hfill
\begin{subfigure}[t]{0.53\columnwidth}
\centering
\begin{tikzpicture}[x=.7cm,y=.7cm]
\path[use as bounding box] (-.3,-1.55) rectangle (5.5,2.3);
\node[obs={Z}] at (0,0) {};
\node[obs={X}] at (1.8,0) {};
\node[obs={A}] at (3.5,.75) {};
\node[obs={B}] at (3.5,-.75) {};
\node[obs={Y}] at (5.2,0) {};
\draw[causal,blue] (Z) -- (X);
\draw[causal] (X) -- (A);
\draw[causal] (X) -- (B);
\draw[causal] (A) -- (Y);
\draw[causal] (B) -- (Y);
\draw[bidirected, bend right=45] (Z) to (X);
\draw[bidirected, bend left=25] (Z) to (A);
\draw[bidirected, bend right=45] (Z) to (B);
\draw[bidirected, bend left=45] (Z.north) to (Y.north);
\draw[bidirected] (X) -- (Y);
\end{tikzpicture}
\caption{}
\label{fig:running-example-hcid}
\end{subfigure}
\caption{Edges shown in blue are those whose expansion identifies $p(y|\doo(x))$. Theorems~\ref{thm:bidirected_paths}, \ref{thm:directed_child}, and~\ref{thm:one-bidirected-path} fully resolve (a), but not (b).}
\label{fig:running-example}
\end{figure}

\begin{thmrep}\label{thm:one-bidirected-path}
Suppose that there exists a bidirected path $\pi$ between $X_i\in\+ X$ and one of its children $C$ in $\cG[\{X_i\}\cup(\An(\+ Y)_{\cG_{\bar{\+ X}}}\setminus\+ X)]$. If $p(\+ y|\doo(\+ x))$ is identifiable in the extended graph $\cG_e$, then $e\in E(\pi)\cup\{(X_i,C)\}$.
\end{thmrep}
\begin{proof}
Suppose $e\notin E(\pi)\cup\{(X_i,C)\}$. Then the path $\pi$ and the edge $X_i\rightarrow C$ remain in $\cG_e$. Let $\cF$ consist of the nodes and bidirected edges of $\pi$ together with $X_i\rightarrow C$, and let $\cF'=\cF[V(\pi)\setminus\{X_i\}]$. Then $\cF$ and $\cF'$ are C-forests with root set $V(\pi)\setminus\{X_i\}$, $X_i\in V(\cF)$, and $\+ X\cap V(\cF')=\emptyset$. This root set lies in $\An(\+ Y)_{\cG_{\bar{\+ X}}}$ by assumption, and hence in $\An(\+ Y)_{\cG_{e,\bar{\+ X}}}$, because expanding an edge preserves directed paths between original nodes. Thus $(\cF,\cF')$ is a hedge for $p(\+ y|\doo(\+ x))$ in $\cG_e$, and $p(\+ y|\doo(\+ x))$ is not identifiable in $\cG_e$ by Theorem~\ref{thm:hedge_criterion}.
\end{proof}

The theorem says that whenever such bidirected paths between $X$ and its children exist, we can restrict attention to those paths and the corresponding directed edges. As a corollary, if two children have such paths that share no edges, expanding one edge will always leave at least one obstruction intact.

\begin{correp}\label{thm:twobidirected}
Suppose that there exist two bidirected paths, $\pi_1$ and $\pi_2$, between $X$ and two of its children, $C_1$ and $C_2$, respectively, in $\cG[\{X\}\cup\An(\+ Y)_{\cG_{\bar X}}]$ such that $E(\pi_1)\cap E(\pi_2)=\emptyset$. Then $p(\+ y|\doo(x))$ is not identifiable in any extended graph of $\cG$.
\end{correp}
\begin{proof}
By Theorem~\ref{thm:one-bidirected-path}, any ID-enabling edge must lie in both $E(\pi_1)\cup\{(X,C_1)\}$ and $E(\pi_2)\cup\{(X,C_2)\}$. These sets are disjoint.
\end{proof}

For example, suppose $X$ has two children that are ancestors of $\+ Y$, each connected to $X$ by a bidirected edge. No single edge expansion can eliminate both obstructions to identification.

The combination of Theorems~\ref{thm:bidirected_paths}, \ref{thm:directed_child}, and~\ref{thm:one-bidirected-path} provides simple yet powerful graphical criteria to find and rule out ID-enabling edges. To illustrate, consider Figure~\ref{fig:running-example}(a). Theorems~\ref{thm:bidirected_paths} and~\ref{thm:directed_child} allow us to find the edges $X\rightarrow W$, $X\leftrightarrow Z$, and $Z\leftrightarrow W$, while Theorem~\ref{thm:one-bidirected-path} rules out all the rest.
In fact, in our exhaustive coverage experiments (Section~\ref{sec:criteria-coverage}), these criteria fully resolve $\geq 90\%$ of nonidentified queries across graphs of 2 to 6 nodes. In our random samples of large graphs, this number reached 92\%--100\% of the sampled cases. However, while powerful, these graphical conditions are not complete; Figure~\ref{fig:running-example}(b) shows an example in which the theorems neither find nor exclude any edges. This motivates our developments in the following sections.
\section{Hedges and the ID algorithm}

The identifiability of a causal effect is related to the structure of the underlying causal graph. In particular, certain graphical structures characterize situations in which a causal effect cannot be identified from the observed data distribution. Hedges are one such graphical structure \citep{shpitserpearl:2006}.

We first introduce the notions of a \textit{forest}, a \textit{C-component}, and a \textit{C-forest}. A \textit{forest} is a graph in which each node has at most one child, and a \textit{C-component} is a graph in which every pair of observed nodes is connected by a bidirected path. A graph that is both a forest and a C-component is called a C-forest. It is often useful to identify the set of nodes in a graph that have no children, which we refer to as the \textit{root set} of the graph. Accordingly, a C-forest whose root set is $\+ R$ is called an $\+ R$-rooted C-forest. We can now use these concepts to define a hedge.

\begin{definition}[Hedge]
Let $\cG = (\+ V, \+ D, \+ B)$ and let $\+ X, \+ Y \subset \+ V$, $\+ X \cap \+ Y = \emptyset$. Let $\cF$ and $\cF'$ be $\+ R$-rooted C-forests such that $\+ X \cap V(\cF) \neq \emptyset$, $\+ X \cap V(\cF') = \emptyset$, $\cF' \subset \cF \subseteq \cG$, and $\+ R \subseteq \An(\+ Y)_{\cG_{\bar{\+ X}}}$. Then the pair $(\cF, \cF')$ forms an $\+R$-rooted hedge for $p(\+ y|\doo(\+ x))$ in $\cG$.
\end{definition}

Since the root set of $\cF$ and $\cF'$ plays an important role in hedge construction, we also include the root set explicitly in the hedge definition.

The starting point for finding ID-enabling edges is a causal graph in which the causal effect of interest is not identifiable. In such a graph, nonidentifiability can be characterized both graphically through the existence of a hedge and algorithmically through the failure of the ID algorithm. In particular, \citet{shpitserpearl:2006} established that a causal effect is nonidentifiable if and only if the graph contains a hedge for the causal effect; a condition that is equivalent to the failure of the ID algorithm. These results were later clarified by \citet{Shpitser:2023}. The main identification results are stated formally below.

\begin{thm}[\citet{Shpitser:2023}, Proposition~1, Condition~1] \label{thm:hedge_criterion}
  The causal effect $p(\+ y|\doo(\+ x))$ is not identifiable in $\cG$ if and only if there exists a hedge for $p(\+ y|\doo(\+ x))$ in $\cG$.
\end{thm}
\begin{cor} \label{thm:not-id-when-hedge}
  The causal effect $p(\+ y|\doo(\+ x))$ is not identifiable in $\cG$ if and only if the ID algorithm fails with inputs $(\+ Y, \+ X, p(\+ v), \cG)$.
\end{cor}

Our method to reduce the set of potentially ID-enabling edges relies on the existence of certain types of hedges in the graph. However, there is no systematic way to directly locate these hedge structures in the original graph: the ID algorithm does not return the exact hedge structure upon failure, but a pair of graphs that can be used to construct the hedge. This applies both to the original formulation by \citet{shpitserpearl:2006} and the modern one line formula \citep{richardson2023admg}.

To reduce the set of potentially ID-enabling directed edges, we need to identify a set of nodes that is guaranteed to form the root set of some hedge in the graph. The following lemma provides such a set by showing that any hedge $(\cF, \cF')$ can be transformed into a hedge whose root set is $V(\cF')$.

\begin{lemrep} \label{lem:hedge-from-hedge}
Let $\cH$ and $\cH'$ be graphs obtained from graphs $\cF$ and $\cF'$, respectively, by removing all outgoing edges from nodes in $V(\cF')$. If $(\cF, \cF')$ forms a hedge then the pair $(\cH, \cH')$ forms a $V(\cF')$-rooted hedge.
\end{lemrep}
\begin{proof}
Since the pair $(\cF, \cF')$ forms a hedge, it satisfies the following conditions of a hedge:
\begin{enumerate}
  \item $V(\cF) \cap \+ X \neq \emptyset$, $V(\cF') \cap \+ X =  \emptyset$.
  \item $\cF' \subset \cF$.
  \item $\cF$ and $\cF'$ are C-forests.
  \item $\cF$ and $\cF'$ have the same root set $\+ R$.
  \item  $\+ R \subseteq \An(\+ Y)_{\cG_{\bar{\+ X}}}$.
\end{enumerate}
Next, we show that these properties also hold for $\cH$ and $\cH'$ with root set $V(\cF')$, which proves that $(\cH, \cH')$ forms a $V(\cF')$-rooted hedge.
\begin{enumerate}
  \item Since $V(\cH) = V(\cF)$ and $V(\cH') = V(\cF')$, we have $V(\cH) \cap \+ X \neq \emptyset$ and $V(\cH') \cap \+ X = \emptyset$.

  \item Because $\cF' \subset \cF$ and only edges common to $\cF'$ and $\cF$ are removed from them to obtain $\cH'$ and $\cH$, respectively, it follows that $\cH' \subset \cH$.

  \item Since only directed edges are removed from $\cF$ and $\cF'$, $\cH$ and $\cH'$ are C-forests.

  \item Since $\cF$ and $\cF'$ have the same root set, the root sets of $\cH$ and $\cH'$ can consist only of nodes in $V(\cF')$. Since $\cH$ and $\cH'$ are obtained from $\cF$ and $\cF'$ by removing all outgoing edges from nodes in $V(\cF')$, the nodes in $V(\cF')$ have no children in $\cH$ and $\cH'$. Therefore, the root set of both $\cH$ and $\cH'$ is $V(\cF')$.

  \item  We have $V(\cF') \setminus \+ R\subseteq \An(\+ R)_{\cF'}$, since otherwise $\+ R$ could not be the root set of $\cF'$. Since $\+ X \cap V(\cF') = \emptyset$, we also have $\cF' \subseteq \cG_{\bar{\+ X}}$. Together, these imply that $V(\cF') \setminus \+ R \subseteq \An(\+ R)_{\cG_{\bar{\+ X}}}$. Furthermore, since $\+ R \subseteq \An(\+ Y)_{\cG_{\bar{\+ X}}}$ and $V(\cF') \setminus \+ R \subseteq \An(\+ R)_{\cG_{\bar{\+ X}}}$, it follows that $V(\cF') \setminus \+ R \subseteq \An(\+ Y)_{\cG_{\bar{\+ X}}}$. Finally, combining the above results with $\+ R \cup (V(\cF') \setminus \+ R) = V(\cF')$, we obtain $V(\cF') \subseteq \An(\+ Y)_{\cG_{\bar{\+ X}}}$.
\end{enumerate}
We have now shown that all properties of a hedge hold for $(\cH, \cH')$ with root set $V(\cF')$. Therefore, $(\cH, \cH')$ forms a $V(\cF')$-rooted hedge.
\end{proof}

To simplify the analysis and state the following results more straightforwardly, we first modify the graph pair returned by the ID algorithm. This leads to a graphical structure that we call an \textit{ID-hedge}. We next define an ID-hedge and establish some useful properties of this structure.

\begin{definition}[ID-hedge]
  Let $(\cC, \cC' )$ be a graph pair returned by the ID algorithm upon failure for $p(\+ y|\doo(\+ x))$ in a causal diagram $\cG$ and let $\cS$ and $\cS'$ be the graphs obtained from $\cC$ and $\cC'$, respectively, by removing all outgoing edges from nodes in $V(\cC')$. The graph pair $(\cS, \cS')$ is called an ID-hedge of $p(\+ y|\doo(\+ x))$ in $\cG$.
\end{definition}

Importantly, we can now state the root set of an ID-hedge.
\begin{lemrep} \label{lem:root-of-id-hedge}
  Let $(\cS, \cS')$ be an ID-hedge of $p(\+ y|\doo(\+ x))$ in $\cG$. Then the root set of both $\cS$ and $\cS'$ is $V(\cS')$.
\end{lemrep}
\begin{proof}
  By the hedge criterion, we can form a hedge from the edge subgraphs of $\cC$ and $\cC'$ that are returned by the ID algorithm. Therefore, the root sets of $\cC$ and $\cC'$ can consist only of nodes in $V(\cC')$. Since $\cS$ and $\cS'$ are obtained from $\cC$ and $\cC'$ by removing all outgoing edges of nodes in $V(\cC')$, the nodes in $V(\cC')$ have no children in $\cS$ and $\cS'$. Thus, the root set of both $\cS$ and $\cS'$ is $V(\cC') = V(\cS')$.
\end{proof}

By Lemma~\ref{lem:root-of-id-hedge}, both graphs of an ID-hedge have the same root set consisting of all nodes of the smaller graph. Next, we show that an ID-hedge contains a hedge with the same root set as the ID-hedge.

\begin{lemrep}
\label{lem:hedge-from-id}
  For any ID-hedge $(\cS, \cS')$ of $p(\+ y|\doo(\+ x))$ in $\cG$, there exist edge subgraphs $\cH \subseteq \cS$ and $\cH' \subseteq \cS'$ that form a $V(\cS')$-rooted hedge $(\cH, \cH')$ for $p(\+ y|\doo(\+ x))$ in $\cG$.
\end{lemrep}
\begin{proof}
   By the hedge criterion, we can find edge subgraphs $\cF \subseteq \cC$ and $\cF' \subseteq \cC'$ that form a hedge $(\cF, \cF')$ for $p(\+ y|\doo(\+ x))$. By Lemma~\ref{lem:hedge-from-hedge}, $(\cH, \cH')$ forms a $V(\cF')$-rooted hedge, where $\cH$ and $\cH'$ are obtained from $\cF$ and $\cF'$, respectively, by removing all outgoing edges from nodes in $V(\cF')$. Since $\cS$ and $\cS'$ are obtained from $\cC$ and $\cC'$, respectively, by removing all outgoing edges from nodes in $V(\cS') = V(\cF')$, it follows that $\cH$ and $\cH'$ are edge subgraphs of $\cS$ and $\cS'$, respectively.
\end{proof}

An ID-hedge can be systematically constructed from the graph pair returned by the ID algorithm upon failure and is guaranteed to contain a hedge with known root set. In the next section, we use these properties to determine which edges can be ID-enabling.

\section{Algorithmic Approach to Find All ID-enabling Edges}

We use ID-hedges to characterize the edges that can be expanded to make a nonidentifiable causal effect identifiable. Since nonidentifiability is characterized by the existence of a hedge, a successful extension must eliminate the relevant hedge structures responsible for nonidentifiability. As ID-hedges provide systematically constructible structures that are guaranteed to contain a hedge with a known root set, we study which edges can be expanded to ensure that no hedge with the same root set and the same two node sets remains within the ID-hedge. This leads to the notion of \textit{breaking an ID-hedge}, which we define next.

\begin{definition}[Breaking an ID-hedge] \label{def:breaking-id-hedge}
  An edge $e$ of $\cG$ breaks a $V(\cS')$-rooted ID-hedge $(\cS, \cS')$ of $p(\+ y|\doo(\+ x))$ if there are no edge subgraphs $\cF \subseteq \cS_{-e}$ and $\cF' \subseteq \cS'_{-e}$ that form a $V(\cS')$-rooted hedge $(\cF, \cF')$ for $p(\+ y|\doo(\+ x))$ in $\cG_e$.
\end{definition}

Note that we consider the graphs $\cS_{-e}$ and $\cS_{-e}'$ instead of $\cS_{e}$ and $\cS_{e}'$. Since the new node has no bidirected path to any other node, every C-component containing it is a singleton. The larger forest of a hedge must also contain an original treatment node, so it cannot contain the new node. Since the smaller forest is a subgraph of the larger forest, it cannot contain the new node either.

Breaking an ID-hedge provides a necessary condition for an edge to be ID-enabling. In particular, if an edge $e$ is ID-enabling, no $V(\cS')$-rooted hedge consisting of exactly the same nodes as an original ID-hedge can remain in $\cG_e$. The following theorem formalizes this connection between breaking ID-hedges and ID-enabling edges.

\begin{thmrep} \label{thm:successful-breaks}
  If $p(\+ y|\doo(\+ x))$ is not identifiable in $\cG$ but is identifiable in $\cG_e$ for some edge $e$ of $\cG$, then $e$ breaks every ID-hedge of $p(\+ y|\doo(\+ x))$ in $\cG$.
\end{thmrep}
\begin{proof}
  If $e$ does not break some ID-hedge, then, by Definition~\ref{def:breaking-id-hedge}, there exists a hedge for the causal effect $p(\+ y|\doo(\+ x))$ in $\cG_e$. Therefore, by Theorem~\ref{thm:hedge_criterion}, $p(\+ y|\doo(\+ x))$ is not identifiable in $\cG_e$.
\end{proof}

Definition~\ref{def:breaking-id-hedge} considers hedges for the causal effect of interest in the extended graph $\cG_e$, since our interest is in identification from $p(\+ v,n)$. However, we may simplify the subsequent analysis by considering the causal effect in the original graph $\cG$ instead. The following theorem shows this equivalence.

\begin{thmrep} \label{thm:breaking-id-hedge}
  An edge $e$ of $\cG$ breaks the $V(\cS')$-rooted ID-hedge $(\cS, \cS')$ of $p(\+ y|\doo(\+ x))$ in $\cG$ if and only if there are no edge subgraphs $\cF \subseteq \cS_{-e}$ and $\cF' \subseteq \cS'_{-e}$ that form a $V(\cS')$-rooted hedge $(\cF, \cF')$ for $p(\+ y|\doo(\+ x))$ in $\cG$.
\end{thmrep}
\begin{proof}
  Let $\cF$ and $\cF'$ be arbitrary edge subgraphs of $\cS_{-e}$ and $\cS'_{-e}$, respectively. Since neither $\cF$ nor $\cF'$ contains the edge $e$, both are subgraphs of $\cG_e$. This ensures that $(\cF,\cF')$ can be considered as a potential hedge in $\cG_e$, since a hedge in a graph must itself consist of subgraphs of that graph.

  We consider two cases: (i) $(\cF,\cF')$ forms a hedge in $\cG$, and (ii) $(\cF,\cF')$ does not form a hedge in $\cG$. We show that in both cases the same holds in $\cG_e$, establishing the desired equivalence.

  (i) If $(\cF, \cF')$ forms a hedge for the causal effect $p(\+ y|\doo(\+ x))$ in $\cG$, it satisfies the following properties:
  \begin{enumerate}
    \item $V(\cF) \cap \+ X \neq \emptyset$ and $V(\cF') \cap \+ X = \emptyset$.
    \item $\cF' \subseteq \cF$.
    \item $\cF$ and $\cF'$ are C-forests.
    \item $\cF$ and $\cF'$ have the same root set $\+ R$.
    \item $\mathbf{R} \subseteq \An(\+ Y)_{\cG_{\bar{\+ X}}}$.
  \end{enumerate}
  Properties 1--4 continue to hold in $\cG_e$, since they depend only on $\cF$ and $\cF'$, and neither $\cF$ nor $\cF'$ contains the edge $e$. Therefore, expanding $e$ does not affect these properties. It remains to show that Property 5 also holds.

  If $e=(A, V)$, then the path $A \rightarrow V$ is replaced by $A \rightarrow N \rightarrow V$ in $\cG_e$. Thus, expanding $e$ preserves all existing directed paths while introducing only the new node $N$ along the replaced edge. Similarly, if $e=A \leftrightarrow V$, then it is replaced by $N \rightarrow A$ and $N \rightarrow V$. This introduces two new directed edges but does not remove any existing directed edges or directed paths. Thus, in both cases, expanding $e$ can only increase the ancestor set of $\+ Y$ by the node $N$ when moving from $\cG_{\bar{\+ X}}$ to $\cG_{e, \bar{\+ X}}$. Hence, $\+ R \subseteq \An(\+ Y)_{\cG_{e, \bar{\+ X}}}$. Thus, all properties required for $(\cF, \cF')$ to form a hedge in $\cG$ continue to hold in $\cG_e$. Consequently, $(\cF, \cF')$ forms a hedge for the causal effect $p(\+ y|\doo(\+ x))$ in $\cG_e$.

  (ii) If $(\cF, \cF')$ does not form a hedge for the causal effect $p(\+ y| \doo(\+ x))$ in $\cG$, then at least one of Properties 1--5 does not hold. If one of Properties 1--4 does not hold in $\cG$, then it also does not hold in $\cG_e$, since expanding $e$ does not affect these properties as shown above. If Property 5 does not hold in $\cG$, then it does not hold in $\cG_e$ either. As shown above, expanding $e$ can only add the new node $N$ to the ancestor set of $\+ Y$ when moving from $\cG_{\bar{\+ X}}$ to $\cG_{e, \bar{\+ X}}$. Therefore, since $\+ R \not\subseteq \An(\+ Y)_{\cG_{\bar{\+ X}}}$ and $N \notin \+ R$, we also have $\+ R \not\subseteq \An(\+ Y)_{\cG_{e, \bar{\+ X}}}$. Thus, $(\cF, \cF')$ does not form a hedge for the causal effect $p(\+ y|\doo(\+ x))$ in $\cG_e$ either.

  Therefore, $(\cF, \cF')$ forms a hedge in $\cG$ if and only if it forms a hedge in $\cG_e$. Hence, the condition for breaking an ID-hedge is equivalent whether the hedge is considered in the original graph $\cG$ or in the extended graph $\cG_e$, which proves the claim.
\end{proof}

Unsurprisingly, it suffices to focus our attention on edges that are present in the hedges for the causal effect of interest.
\begin{correp} \label{thm:characterization}
  Let $e$ be an edge of $\cG$. Then $p(\+ y|\doo(\+ x))$ is identifiable in $\cG_e$ if and only if $e$ is in every hedge for $p(\+ y|\doo(\+ x))$ in $\cG$.
\end{correp}
\begin{proof}
  Expanding $e$ preserves ancestral relations among the original nodes, including after intervention on $\+ X$ (see the proof of Theorem~\ref{thm:breaking-id-hedge}). Moreover, $N$ belongs to no hedge, since it has no bidirected edges and the larger forest of a hedge is a C-forest containing a treatment node and at least one other node. Hence the hedges in $\cG_e$ are exactly the hedges in $\cG$ that do not contain $e$, so $\cG_e$ has no hedge if and only if every hedge in $\cG$ contains $e$. The claim then follows from Theorem~\ref{thm:hedge_criterion}.
\end{proof}

We next characterize the individual edges that break an ID-hedge. For a directed edge, the relevant property is whether removing the edge causes its source node to lose all of its children, thereby changing the root set of the ID-hedge. For a bidirected edge, the relevant property is whether removing the edge destroys the bidirected connectivity required for a C-component. Characterizing edges with these properties will allow us to restrict the set of potentially ID-enabling edges. Based on these properties, we define two types of edges that are relevant for breaking an ID-hedge: weak directed edges and weak bidirected edges.

\begin{definition}[Weak Directed Edge]
\label{def:weak-directed-edge}
  Let $\cI = (\cS, \cS')$ be an ID-hedge of $p(\+ y|\doo(\+ x))$ in $\cG$. A directed edge $e = (A, B)$ of $\cS$ is called a \emph{weak directed edge} of $\cI$ if $\mathrm{Ch}(A)_{\cS_{-e}} = \emptyset$. The set of weak directed edges of $\cI$ is denoted by $\cW_d(\cI)$.
\end{definition}

\begin{definition}[Weak Bidirected Edge]
\label{def:weak-bidirected-edge}
  Let $\cI = (\cS, \cS')$ be an ID-hedge of $p(\+ y|\doo(\+ x))$ in $\cG$. A bidirected edge $e$ is called a \emph{weak bidirected edge} of $\cI$ if $\cS_{-e}$ or $\cS_{-e}'$ is not bidirected-connected. The set of weak bidirected edges of $\cI$ is denoted by $\cW_b(\cI)$.
\end{definition}

{The following theorem shows that only these two types of edges can break an ID-hedge.}

\begin{thmrep} \label{thm:weak-edges}
  If an edge $e$ breaks a $V(\cS')$-rooted ID-hedge $(\cS, \cS')$ then $e$ is either a weak directed edge or a weak bidirected edge.
\end{thmrep}
\begin{proof}
If $e\notin E(\cS)$, deleting $e$ changes neither $\cS$ nor $\cS'$. The hedge supplied by Lemma~\ref{lem:hedge-from-id} remains, and Theorem~\ref{thm:breaking-id-hedge} implies that $e$ does not break the ID-hedge. Such an edge is neither weak directed nor weak bidirected. Hence, suppose $e\in E(\cS)$. We show that a non-weak edge cannot break the ID-hedge. We use the following notation: $\cG_{e \to e'}$ denotes the graph obtained from $\cG$ by removing the edge $e$ and adding a new edge $e'$. If $e \notin E(\cG)$, then $\cG_{e \to e'} = \cG$.

Suppose first that $e$ is a directed edge $(A, B)$ but it is not a weak directed edge. By Lemma~\ref{lem:hedge-from-id}, we have edge subgraphs $\cF \subseteq \cS$ and $\cF' \subseteq \cS'$ forming a $V(\cS')$-rooted hedge $( \cF, \cF' )$ for the causal effect $p(\+ y|\doo(\+ x))$ in $\cG$. Since $e$ is not a weak directed edge, $\Ch(A) \neq \emptyset$ in $\cS_{-e}$, and thus there exists an outgoing edge $e'$ from node $A$ in $\cS_{-e}$. Since $\cF$ is a $V(\cS')$-rooted C-forest, it contains one outgoing edge from $A$. If the outgoing edge of $A$ in $\cF$ is not $e$, then the edge subgraphs $\cF \subseteq \cS_{-e}$ and $\cF' \subseteq \cS'_{-e}$ still form a $V(\cS')$-rooted hedge $(\cF, \cF')$ for the causal effect $p(\+ y|\doo(\+ x))$ in $\cG$. If the outgoing edge of $A$ in $\cF$ is $e$, then edge subgraphs $\cF_{e \to e'} \subseteq \cS_{-e}$ and $\cF'_{e \to e'} \subseteq \cS'_{-e}$ form a $V(\cS')$-rooted hedge $(\cF_{e \to e'}, \cF'_{e \to e'} )$ for the causal effect $p(\+ y|\doo(\+ x))$ in $\cG$.

Suppose next that $e$ is a bidirected edge but it is not a weak bidirected edge. Lemma~\ref{lem:hedge-from-id}, we have edge subgraphs $\cF \subseteq \cS$ and $\cF' \subseteq \cS'$ forming a $V(\cS')$-rooted hedge $( \cF, \cF' )$ for the causal effect $p(\+ y|\doo(\+ x))$ in $\cG$. If $\cS_{-e}$ and $\cS_{-e}'$ are bidirected-connected, we can construct edge subgraphs $\cH \subseteq \cS_{-e}$ and $\cH' \subseteq \cS_{-e}'$. The graphs $\cH$ and $\cH'$ have the same directed edges as in $\cF$ and $\cF'$, respectively, and the same bidirected edges as in $\cS_{-e}$ and $\cS_{-e}'$, respectively. Thus, $(\cH, \cH')$ forms a $V(\cS')$-rooted hedge for the causal effect $p(\+ y|\doo(\+ x))$ in $\cG$.

\end{proof}

\begin{algorithm}[t]
\caption{Hedge-Collecting Identification Algorithm (HCID). The inputs are node sets \(\+ X\) and \(\+ Y\), representing the intervention and outcome variables, respectively, and an ADMG $\cG$. The output is a set of graph pairs obtained from encountered ID failures. The collection of node sets of the maximal C-components of $\cG$ is denoted by $C(\cG)$.}
\label{alg:hcid}
\begin{algorithmic}[1]

\Function{HCID}{\(\+ Y, \+ X, \cG\)}

\State $\+ V \gets V(\cG)$
\If{\(\+ X=\emptyset\)}
  \State \textbf{return} \(\emptyset\).
\EndIf

\If{\(\+ V \neq \An(\+ Y)_\cG\)}
  \State \textbf{return} \(\operatorname{HCID}(\+ Y, \+ X \cap \An(\+ Y)_\cG, \cG[\An(\+ Y)_\cG])\).
\EndIf

\State Let \(\+ W = (\+ V \setminus \+ X) \setminus \An(\+ Y)_{\cG_{\bar{\+ X}}}\). If \(\+ W \neq \emptyset\),
\Statex \quad \textbf{return} \(\operatorname{HCID}(\+ Y, \+X \cup \+ W, \cG)\).

\If{\(C(\cG[\+ V \setminus \mathbf{X}]) = \{\+ S_1,\ldots,\+ S_k\}\), \(k\geq 2\)} \label{step:factorization}
  \Statex \quad \textbf{return} \(\bigcup_{i=1}^k \operatorname{HCID}(\+ S_i, \+ V \setminus \+ S_i, \cG)\).
\EndIf

\If{\(C(\cG[\+ V \setminus \mathbf{X}]) = \{\+ S\}\)}
  \If{\(C(\cG)=\{\+ V\}\)}
    \State \textbf{return} \(\{(\cG, \cG[\+ S])\}\).
  \EndIf
  \If{\(\+ S \in C(\cG)\)}
    \State \textbf{return} \(\emptyset\).
  \EndIf
  \If{\((\exists \+ S')\, \+ S \subset \+ S'\in C(\cG)\)}
    \State \textbf{return} \(\operatorname{HCID}(\+ Y, \+ X \cap \+ S', \cG[\+ S'])\).
  \EndIf
\EndIf
\EndFunction
\end{algorithmic}
\end{algorithm}

We have now developed a theoretical framework for reducing the set of potentially ID-enabling edges by focusing on weak directed and weak bidirected edges that can break an ID-hedge. The next step is to obtain all ID-hedges for which we leverage the ID algorithm. In step~4 of the standard ID algorithm, the graph is decomposed into its C-components and a recursive ID call is made for each component if there are at least two C-components. Since the C-components are not ordered, there is no fixed order in which the corresponding recursive branches must be evaluated. Consequently, when failures occur in multiple branches, different orderings may cause different failing branches to be encountered and propagated first. This provides a way to obtain multiple graph pairs associated with ID failures.

To exploit this observation, we introduce a modified version of the ID algorithm, which we call the Hedge-Collecting Identification algorithm (HCID, Algorithm~\ref{alg:hcid}). HCID preserves the graphical recursion of the standard ID algorithm but modifies its failure-handling behavior. Whereas a failure in the standard ID algorithm propagates through the recursive calls and terminates the identification procedure, a failing branch in HCID returns a singleton set containing the corresponding graph pair. A branch that terminates without failure returns the empty set, and when the recursion branches into multiple C-components on line~\ref{step:factorization}, HCID returns the union of the graph-pair sets obtained from the individual branches. Thus, a failure in one branch does not prevent the remaining branches from being processed.

As in the standard ID setting, HCID assumes that the observational joint distribution is known. However, HCID does not explicitly require this distribution as an input, since it does not construct the identifying functional. The parts of the standard ID algorithm concerned solely with constructing the identifying functional are therefore omitted. Instead, HCID records the graph pair associated with each encountered failure condition. Consequently, if $h$ failure conditions are encountered, HCID returns the collection $\mathcal{P} = \{(\cG_1,\cG'_1),\ldots,(\cG_h,\cG'_h)\},$ where $\mathcal{P}$ is empty if and only if the causal effect is identifiable in $\cG$.

The construction of HCID is based on the idea that each graph pair it returns corresponds to a failure that could be encountered by the standard ID algorithm under a suitable ordering of its recursive branches. We formalize this connection in the following theorem. This ensures that the graph pairs collected by HCID have the same interpretation as failure pairs returned by the standard ID algorithm.

\begin{lemrep} \label{lem:ordering}
For every graph pair $(\cC_i,\cC'_i)$ returned by HCID, there exists an ordering of the recursive branches of the standard ID algorithm under which ID returns $(\cC_i,\cC'_i)$ as its failure pair.
\end{lemrep}

\begin{proof}
Let $(\cC_i,\cC'_i) \in \mathcal{P}$ be an arbitrary graph pair returned by HCID. The pair is returned when a recursive branch reaches the same failure condition as in the standard ID algorithm. Furthermore, HCID follows the same graphical recursion as the standard ID algorithm within each recursive branch.

The standard ID algorithm decomposes the graph into its C-components and makes a recursive call for each of them. However, the algorithm does not specify an ordering of these C-components. Therefore, whenever multiple C-components are encountered on the recursive path leading to $(\cC_i,\cC'_i)$, the C-component containing the continuation of this path may be considered first.

Choosing this ordering at every such branching step causes the standard ID algorithm to follow the same sequence of recursive calls as HCID along the path leading to $(\cC_i,\cC'_i)$. The recursion therefore reaches the same failure condition, and the resulting failure is propagated by the standard ID algorithm. Hence, ID returns $(\cC_i,\cC'_i)$ as its failure pair.

Since $i$ was arbitrary, the result holds for every graph pair in $\mathcal{P}$.
\end{proof}

Since every graph pair in $\mathcal{P}$ can also be returned by the standard ID algorithm under an appropriate ordering of its recursive branches, the ID-hedge construction introduced previously can be applied to each graph pair in $\mathcal{P}$. Consequently, instead of constructing an ID-hedge from only the first failure encountered by ID, the graph pairs collected by HCID can be used to obtain multiple ID-hedges.

For each resulting ID-hedge $\cI$, let $\+S_i$ denote the set of edges that break $\cI$. If an edge $e$ is ID-enabling, then $e$ must break every ID-hedge obtained from the graph pairs in $\mathcal{P}$. Therefore, any such edge must satisfy
$e \in \bigcap_{i=1}^{h} \+ S_i$. Hence, the set $\+ C = \bigcap_{i=1}^{h} \+ S_i$ forms a set of candidate edges containing all potentially ID-enabling edges. Edges outside $\+ C$ cannot render the causal effect identifiable, since at least one of the collected ID-hedges would remain unbroken in the corresponding extended graph. The extended graphs corresponding to edges in $\+ C$ can then be tested individually for identifiability by the ID algorithm.

We now formulate an algorithm for finding all ID-enabling edges as Algorithm~\ref{alg:successful-edges}. The algorithm collects ID-hedges with HCID and intersects their weak edges, yielding at most $2n-2$ candidates---and, as we show below, this bound is sharp (Theorem~\ref{thm:measurement-complexity}). It then tests each candidate's expansion with ID. Note that this is not the only possible algorithm, and that different necessary conditions can be combined to restrict the search without affecting correctness (e.g., Theorem~\ref{thm:one-bidirected-path}). Further necessary conditions can reduce the candidate set, while fewer conditions preserve correctness but may require more ID tests. Appendix~\ref{app:variants} examines several variants.
\par

\begin{algorithm}[t]
\caption{Algorithm for finding ID-enabling edges. The inputs are an ADMG $\cG$ in which the causal effect $p(\+ y|\doo(\+ x))$ is not identifiable, an intervention set $\+ X$, and an outcome set $\+ Y$. The output is the set of ID-enabling edges $\+ S$ such that $p(\+ y|\doo(\+ x))$ is identifiable in $\cG_e$ for all $e \in \+ S$. $\Call{ID}{}$ denotes a variation of the ID algorithm that omits the construction of the identifying functional analogous to $\Call{HCID}{}$.}
\label{alg:successful-edges}
\begin{algorithmic}[1]
\Function{ID-EnablingEdges}{\(\+ Y, \+ X, \cG\)}
\State $\cP \gets \Call{HCID}{\mathbf{Y},\mathbf{X}, \cG}$ \label{step:hcid}
\For{$(\cC_i,\cC'_i) \in \cP$}
  \State $\+ Z \gets V(\cC'_i)$ \label{step:smaller_nodes}
  \State $(\cC_i, \cC'_i) \gets (\cC_{i, \underline{\+ Z}}, \cC'_{i, \underline{\+ Z}})$ \label{step:remove_incoming}
\EndFor
\State $\+ C \gets \displaystyle\bigcap_{i=1}^{|\cP|} (\cW_d(\cC_i,\cC'_i) \cup \cW_b(\cC_i,\cC'_i))$ \label{step:intersect}
\State $\+ S \gets \emptyset$
\For{$e \in \+ C$}
    \If{$\Call{ID}{\+ Y, \+ X, \cG_e}$ does not fail} \label{step:id_check}
        \State $\+ S \gets \+ S \cup \{e\}$
    \EndIf
\EndFor
\State \Return $\+ S$ \label{step:successful}
\EndFunction
\end{algorithmic}
\end{algorithm}

\begin{thmrep}[Soundness of Algorithm~\ref{alg:successful-edges}]
  Let $\cG$ be a causal diagram such that $p(\+ y|\doo(\+ x))$ is not identifiable in $\cG$ and let $\+ S = \Call{ID-EnablingEdges}{\+ Y, \+ X, \cG}$. Then $p(\+ y|\doo(\+x))$ is identifiable in $\cG_e$ for every $e \in \+ S$.
\end{thmrep}
\begin{proof}
  The identifiability of $p(\+ y|\doo(\+ x))$ in $\cG_e$ is verified by the ID algorithm for each member $e \in \+ S$ on line~\ref{step:id_check}. The claim now follows from the soundness of ID.
\end{proof}

\begin{thmrep}[Completeness of Algorithm~\ref{alg:successful-edges}]
\label{thm:successful-edges-completeness}
  Let $\cG$ be a causal diagram such that $p(\+ y|\doo(\+ x))$ is not identifiable in $\cG$ and let $\+ S = \Call{ID-EnablingEdges}{\+ Y, \+ X, \cG}$. If there exists an edge $e$ of $\cG$ such that $p(\+ y|\doo(\+ x))$ is identifiable in $\cG_e$, then $e \in \+ S$.
\end{thmrep}
\begin{proof}
  Let $e$ be any ID-enabling edge. Since the original effect is not identifiable, HCID returns at least one pair, and every pair it collects can also be returned by ID under a suitable ordering of the branches, as shown by Lemma~\ref{lem:ordering}. Each corresponding pair $(\cS_i,\cS_i')$ is therefore an ID-hedge for the original effect. By Theorem~\ref{thm:successful-breaks}, $e$ must break each such ID-hedge, and by Theorem~\ref{thm:weak-edges}, $e \in \cW_d(\cS_i,\cS_i') \cup \cW_b(\cS_i,\cS_i')$ for every $i$. A directed edge can only belong to the directed sets and a bidirected edge can only belong to the bidirected sets. Thus $e$ belongs to the candidate set $\+ C$ in the algorithm. Completeness of ID ensures that its final test succeeds, so the algorithm returns $e$.
\end{proof}

We demonstrate the operation of Algorithm~\ref{alg:successful-edges} via an
example. Consider the graph $\cG$ in Figure~\ref{fig:running-example}(b)
where $p(y|\doo(x))$ is not identifiable. Note that none of the conditions introduced in Section~\ref{sec:graphical-criteria} reduce the number of edges to be tested.
Calling $\Call{HCID}{Y,X,\cG}$ on line~\ref{step:hcid} produces three graph
pairs: $(\cG[\{Z,X,A\}], \cG[A])$, $(\cG[\{Z,X,B\}], \cG[B])$, and
$(\cG, \cG[Y])$. The corresponding root sets are $\{A\}$, $\{B\}$, and
$\{Y\}$. As each root has no children in either graph of its pair, no
changes are made on line~\ref{step:remove_incoming}.
Next, the set of candidates $\+ C$ is constructed on
line~\ref{step:intersect}. For the first two pairs, the common weak edges
are $(Z,X)$ and $\{Z,X\}$. For the third pair, the weak directed edges are
$(Z,X)$, $(A,Y)$, and $(B,Y)$, and the weak bidirected edges are
$\{Z,A\}$ and $\{Z,B\}$. This reduces the set to a single directed edge $\+ C = \{(Z,X)\}$.
The ID test on line~\ref{step:id_check} determines that $p(y|\doo(x))$
is identifiable after expanding $(Z,X)$, so
$\+ S = \+ C = \{(Z,X)\}$ on line~\ref{step:successful}.
Appendix~\ref{app:algorithm-walkthrough} illustrates the operation of Algorithm~\ref{alg:successful-edges} on  Figure~\ref{fig:running-example}(a).

The computational complexity of Algorithm~\ref{alg:successful-edges} is outlined by the following theorem.

\begin{thmrep}[Complexity of Algorithm~\ref{alg:successful-edges}]
\label{thm:measurement-complexity}
Let $n=|\+ V|$, $m=|\+ D|+|\+ B|$, and let $q$ be the number of candidates selected by Algorithm~\ref{alg:successful-edges}. Its output $\+ S$ satisfies $|\+ S|\le q\le\min\{m,2n-2\}$. The bound $2n-2$ is sharp for both candidates and ID-enabling edges. With adjacency lists and graphical ID decisions, the algorithm runs in $O((q+1)n^2(n+m))$, and hence $O(n^3(n+m))$, time. Processing each pair as it is returned permits $O(n+m)$ working space.
\end{thmrep}
\begin{proof}
We first bound the number of candidates. Consider any ID-hedge $(\cS,\cS')$ constructed by Algorithm~\ref{alg:successful-edges}, which exists because the input query is not identifiable. By Lemma~\ref{lem:root-of-id-hedge} and Definition~\ref{def:weak-directed-edge}, there are at most $n-1$ weak directed edges. Similarly, by Definition~\ref{def:weak-bidirected-edge}, the weak bidirected edges cannot form a cycle, because $\cS'$ is a subgraph of $\cS$. There are therefore at most $n-1$ weak bidirected edges.

Algorithm~\ref{alg:successful-edges} selects a subset of these weak edges. Since its output is a subset of the candidates,
\[
|\+ S|\le q\le \min\{m,2n-2\}.
\]

To prove sharpness, for $n\ge2$, consider the graph consisting of the directed chain
\[
Z_1\rightarrow \cdots \rightarrow Z_{n-2}\rightarrow X\rightarrow Y
\]
and the bidirected path
\[
X\leftrightarrow Z_1\leftrightarrow\cdots
\leftrightarrow Z_{n-2}\leftrightarrow Y,
\]
where the sequence of $Z$ nodes is empty when $n=2$. The full graph and its singleton subgraph on $Y$ form a hedge, so the query $p(y|\doo(x))$ is not identifiable by Theorem~\ref{thm:hedge_criterion}.

Every hedge for this query has root set $\{Y\}$, since $\An(Y)_{\cG_{\bar{X}}}=\{X,Y\}$ and the smaller C-forest excludes $X$. Its larger C-forest must therefore contain the entire bidirected path connecting $X$ to $Y$, and the entire directed chain leading to the root $Y$. Thus, every edge belongs to every hedge. By Corollary~\ref{thm:characterization}, all $2n-2$ edges are ID-enabling, and Theorem~\ref{thm:successful-edges-completeness} ensures that Algorithm~\ref{alg:successful-edges} returns them all. This establishes that the bound is sharp for both the number of candidates and the number of ID-enabling edges.

We now bound the running time. HCID branches only over disjoint C-components, so it has at most $n$ terminal branches. Each recursive call along a branch either removes nodes or adds them to the intervention set, so a branch contains $O(n)$ calls, each taking $O(n+m)$ time. HCID therefore runs in $O(n^2(n+m))$ time and returns at most $n$ graph pairs. The same bound holds for ID when the identifying functional is not constructed.

We compute the candidate intersection one ID-hedge at a time. By Definitions~\ref{def:weak-directed-edge} and~\ref{def:weak-bidirected-edge}, we can check whether an edge is weak in $O(n+m)$ time. The first ID-hedge requires at most $m$ such checks, and by the bound above, each subsequent ID-hedge requires at most $2n-2$. Since there are at most $n$ ID-hedges and $m\le n(n-1)$, constructing the ID-hedges and intersecting their weak edges takes $O(n^2(n+m))$ time.

By Definition~\ref{def:extended_graph}, each extended graph has $n+1$ nodes and $m+1$ edges. Each of the $q$ final ID tests therefore takes $O(n^2(n+m))$ time. Adding these costs gives $O((q+1)n^2(n+m))$. Since $q\le 2n-2$, this is $O(n^3(n+m))$.

For the space bound, process the C-components one at a time and replace the current graph at each reduction. Their disjoint node sets require $O(n)$ space. We keep the candidates that are weak edges of the current ID-hedge, then remove it from memory. The extended graphs are then tested one at a time. At any point, we need only store a constant number of graphs and the arrays used for searches, which require $O(n+m)$ space. Since the candidate and output sets use $O(n)$ space, the total space needed is $O(n+m)$.
\end{proof}

\begin{toappendix}
\section{Algorithm walkthrough: Figure~\ref{fig:running-example}(a)}
\label{app:algorithm-walkthrough}

Consider the graph $\cG$ in Figure~\ref{fig:running-example}(a) where $p(y|\doo(x))$ is not identifiable. Calling $\Call{HCID}{Y,X,\cG}$ on line~\ref{step:hcid} produces two graph pairs: $(\cG, \cG[Y])$ and $(\cG[\{X,Z,W\}], \cG[\{Z, W\}])$. For the first pair, we have $\+ Z = \{Y\}$ on line~\ref{step:smaller_nodes}. As $Y$ has no children in either graph of this pair, no changes to the graphs are made on line~\ref{step:remove_incoming}. For the second pair, we have that $\+ Z = \{Z, W\}$, so the directed edge $(Z, W)$ is removed from both graphs of the pair.
Next, the set of candidates $\+ C$ is constructed on line~\ref{step:intersect}. For the first pair, $X$ and $W$ both only have one child in $\cG$, so the weak directed edges are $(X, W)$ and $(W, Y)$. All bidirected edges of $\cG$ are weak, i.e., $\{X, Z\}, \{X, Y\}$ and $\{Z, W\}$, as $\cG$ is no longer bidirected connected if any such edge is removed. For the second pair, the only weak directed edge is $(X, W)$ and the weak bidirected edges are $\{X, Z\}$ and $\{Z, W\}$. Thus $\+ C = \{(X, W), \{X, Z\}, \{Z, W\}\}.$ The ID test on line~\ref{step:id_check} determines that $p(y|\doo(x))$ is identifiable in the extended graph of all three edges, so $\+ S = \+ C$ in this instance on line~\ref{step:successful}.

\end{toappendix}

\section{Computational experiments}
\label{sec:computational-experiments}
In this section, we first study how often expanding a single edge enables identification through exhaustive experiments on small graphs. Next, we evaluate the coverage of the graphical criteria of Section~\ref{sec:graphical-criteria} on these graphs and as well as on random samples of larger graphs. Finally, we compare the runtime of Algorithm~\ref{alg:successful-edges} against expanding and testing every edge. Details of the experiments are provided in Appendix~\ref{app:computational-experiments}.

\subsection{Identification via edge expansion}

\begin{table}[t]
\centering\small
\setlength{\tabcolsep}{3pt}
\renewcommand{\arraystretch}{1.12}
\begin{tabular*}{\linewidth}{@{\extracolsep{\fill}}rrrrrr@{}}
\hline
$n$ & Non-ID & \multicolumn{4}{c}{ID-enabling (\%)} \\
\cline{3-6}
 & & None & D & B & Both \\
\hline
2 & 1 & 0.0 & 0.0 & 0.0 & 100.0 \\
3 & 26 & 7.7 & 15.4 & 7.7 & 69.2 \\
4 & 2{,}010 & 19.3 & 24.2 & 11.3 & 45.2 \\
5 & 513{,}602 & 32.2 & 29.0 & 11.3 & 27.4 \\
6 & 468{,}747{,}580 & 44.8 & 31.0 & 9.1 & 15.2 \\
\hline
\end{tabular*}
\caption{ID-enabling edge expansions among nonidentified queries (\%). D: directed only; B: bidirected only; Both: each type succeeds in a separate expansion.
}
\label{tab:exhaustive-results}
\end{table}

We consider every case in which $p(y|\doo(x))$ is not identifiable in ADMGs with two through six nodes. We consider every choice of distinct $X$ and $Y$ such that all nodes are ancestors of $Y$, up to relabeling that preserves the treatment and outcome. These experiments cover more than $469$ million nonidentified cases.

Table~\ref{tab:exhaustive-results} shows how often identification can be enabled by expanding a directed edge, a bidirected edge, either type, or neither. Expanding a single edge enables identification in $100\%$ of cases with two nodes, $92.3\%$ with three, $80.7\%$ with four, $67.8\%$ with five, and $55.2\%$ with six. Thus, a single expansion suffices in more than half of the cases at every size considered.

At six nodes, $31.0\%$ of cases can be solved only by expanding a directed edge, $9.1\%$ only by expanding a bidirected edge, and $15.2\%$ by expanding an edge of either type. These results illustrate the value of considering both causal mechanisms and sources of confounding when seeking additional observations.

With two variables, successful expansions yield exactly two identification strategies: backdoor and frontdoor adjustment, as depicted in Figure~\ref{fig:simple_example}. With three variables, successful expansions yield four strategies: backdoor, frontdoor, the Napkin (trapdoor) strategy \citep{pearl:2018,helske2021trapdoor,guo2025napkin}, and a fourth strategy based on frontdoor identification of an intermediate effect---see Figure~\ref{fig:identification-strategies}.

\subsection{Coverage of the graphical criteria}
\label{sec:criteria-coverage}
We evaluate how often Theorems~\ref{thm:bidirected_paths}, \ref{thm:directed_child}, and~\ref{thm:one-bidirected-path} fully resolve a case by determining all edges---if any---that are ID-enabling. We also measure how often Theorem~\ref{thm:bidirected_paths} or~\ref{thm:directed_child} finds an ID-enabling edge whenever one exists. We apply these theorems to the exhaustive cases and to random samples of sparse and dense graphs with $n\in\{10,20,40,60,80,160\}$ nodes. Directed and bidirected edge probabilities are $4/n$ and $2/n$ for sparse graphs, and $0.30$ and $0.20$ for dense graphs. To consider different edge densities, we further examine $25$ combinations of directed and bidirected edge probabilities at $n=40$.

In the exhaustive experiments, Theorems~\ref{thm:bidirected_paths}, \ref{thm:directed_child}, and~\ref{thm:one-bidirected-path} fully resolve $90$--$100\%$ of cases, and Theorems~\ref{thm:bidirected_paths} and~\ref{thm:directed_child} find an ID-enabling edge in at least $99.96\%$ of cases whenever a solution exists. Across all settings for larger graphs, the corresponding percentages are $87.70$--$100\%$ and at least $99.92\%$. Detailed results are given in Tables~\ref{tab:criteria-coverage} and~\ref{tab:density-coverage} of Appendix~\ref{app:computational-experiments}.

\begin{figure}[t]
\centering
\begin{subfigure}[t]{0.48\linewidth}
\centering
\begin{tikzpicture}[every label/.append style={text=black}]
  \path[use as bounding box] (-0.25,-0.65) rectangle (2.8,1.2);
  \node[obs={Z}] at (0,0) {};
  \node[obs={N},draw=blue] at (0.85,0) {};
  \node[obs={X}] at (1.7,0) {};
  \node[obs={Y}] at (2.55,0) {};
  \draw[causal,blue] (Z) -- (N);
  \draw[causal,blue] (N) -- (X);
  \draw[causal] (X) -- (Y);
  \draw[bidirected,bend right=45] (Z) to (X);
  \draw[bidirected,bend left=45] (Z) to (Y);
\end{tikzpicture}
\caption{Napkin (trapdoor)}
\end{subfigure}\hfill
\begin{subfigure}[t]{0.48\linewidth}
\centering
\begin{tikzpicture}[every label/.append style={text=black}]
  \path[use as bounding box] (-0.25,-0.65) rectangle (2.8,1.2);
  \node[obs={X}] at (0,0) {};
  \node[obs={N},draw=blue] at (0.85,0) {};
  \node[obs={Z}] at (1.7,0) {};
  \node[obs={Y}] at (2.55,0) {};
  \draw[causal,blue] (X) -- (N);
  \draw[causal,blue] (N) -- (Z);
  \draw[causal] (Z) -- (Y);
  \draw[causal,bend left=45] (X) to (Y);
  \draw[bidirected,bend right=45] (X) to (Z);
\end{tikzpicture}
\caption{Intermediate frontdoor}
\end{subfigure}
\caption{Representative expanded graphs for the two additional identification strategies in 3-variable models. The newly observed variable $N$ and the edges introduced by expansion are highlighted in blue.
}
\label{fig:identification-strategies}
\end{figure}
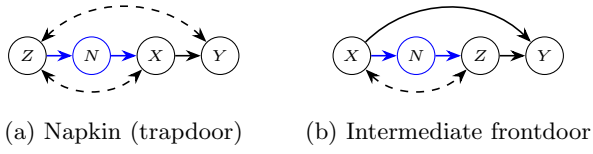

\subsection{Runtime comparisons}
Finally, using the same samples of larger graphs, we compare the runtime of Algorithm~\ref{alg:successful-edges} with the brute force procedure of expanding and testing every edge. For each case, speedup is the runtime of the brute force procedure divided by that of Algorithm~\ref{alg:successful-edges}.

Figure~\ref{fig:measurement-runtime} shows the runtime median and interquartile range for sparse and dense graphs. At $160$ nodes, median speedups are $5.7$ for sparse graphs and $1{,}896.0$ for dense graphs. Among cases where an ID-enabling edge exists, the corresponding speedups are $5.7$ and $1{,}348.1$. Results  across the $25$ combinations of edge probabilities at $n=40$ are similar and defered to the appendix.

\begin{toappendix}
\section{Computational experiments}
\label{app:computational-experiments}

\subsection{Exhaustive experiments}
\label{app:exhaustive-experiments}
\label{app:criteria-coverage}
We start with a set of exhaustive experiments, covering every case in which $p(y|\doo(x))$ is not identifiable in ADMGs with two through six nodes. We consider every choice of distinct $X$ and $Y$ such that all nodes are ancestors of $Y$, up to relabeling that preserves the treatment and the outcome.

For each of these cases, we run Algorithm~\ref{alg:successful-edges} and compare the edges it returns with those found by a brute force procedure that expands each edge separately and applies ID to the resulting graph. As expected, both procedures return the same set of edges, corroborating the correctness of our algorithm.

We also count the number of cases where identification can be enabled only by expanding a directed edge, only by expanding a bidirected edge, by either type, or by neither type. These results are shown in Table~\ref{tab:exp-opportunities}. Expanding a single edge enables identification in $100\%$ of the cases with two nodes, $92.3\%$ with three, $80.7\%$ with four, $67.8\%$ with five, and $55.2\%$ with six. At six nodes, $31.0\%$ of cases can be solved only by expanding a directed edge, $9.1\%$ only by expanding a bidirected edge, and $15.2\%$ by expanding an edge of either type.

\begin{table}[htbp]
\centering\small
\setlength{\tabcolsep}{4pt}
\begin{tabular*}{\linewidth}{@{\extracolsep{\fill}}rrrrrrr@{}}
\hline
$n$ & Total queries & Nonidentified & Neither & Directed only & Bidirected only & Both \\
\hline
2 & 2 & 1 & 0 & 0 & 0 & 1 \\
3 & 40 & 26 & 2 & 4 & 2 & 18 \\
4 & 2{,}568 & 2{,}010 & 387 & 487 & 227 & 909 \\
5 & 581{,}424 & 513{,}602 & 165{,}443 & 149{,}132 & 58{,}282 & 140{,}745 \\
6 & 496{,}583{,}784 & 468{,}747{,}580 & 209{,}916{,}840 & 145{,}086{,}131 & 42{,}537{,}403 & 71{,}207{,}206 \\
\hline
\end{tabular*}
\caption{Exhaustive classification of nonidentified queries. ``Both'' means that an edge of each type succeeds when expanded individually. The number of nodes $n$ refers to the original graph.}
\label{tab:exp-opportunities}
\end{table}

We also apply Theorems~\ref{thm:bidirected_paths}, \ref{thm:directed_child}, and~\ref{thm:one-bidirected-path} to the same cases, to examine the coverage of the graphical criteria developed in Section~\ref{sec:graphical-criteria}. We measure how often these theorems fully resolve a case by determining all edges---if any---that are ID-enabling. We also measure how often Theorem~\ref{thm:bidirected_paths} or~\ref{thm:directed_child} finds an ID-enabling edge whenever one exists. Table~\ref{tab:criteria-coverage} reports the results. Theorems~\ref{thm:bidirected_paths}, \ref{thm:directed_child}, and~\ref{thm:one-bidirected-path} fully resolve $90$--$100\%$ of the nonidentified cases, and Theorems~\ref{thm:bidirected_paths} and~\ref{thm:directed_child} find an ID-enabling edge in at least $99.96\%$ of cases, whenever a solution exists.

\begin{table}[htbp]
\centering
\setlength{\tabcolsep}{6pt}
\renewcommand{\arraystretch}{1.12}
\begin{tabular*}{\linewidth}{@{}r@{\extracolsep{\fill}}rrrr@{}}
\toprule
Nodes & \shortstack{Nonidentified\\queries} & \shortstack{Queries admitting\\an expansion} & \shortstack{Fully\\resolved (\%)} & \shortstack{Found an\\expansion (\%)} \\
\midrule
\multicolumn{5}{@{}l}{\emph{Exhaustive experiments}} \\
2 & 1 & 1 & 100.00 & 100.00 \\
3 & 26 & 24 & 92.31 & 100.00 \\
4 & 2{,}010 & 1{,}623 & 90.00 & 100.00 \\
5 & 513{,}602 & 348{,}159 & 90.70 & 99.97 \\
6 & 468{,}747{,}580 & 258{,}830{,}740 & 92.83 & 99.96 \\
\addlinespace[4pt]
\multicolumn{5}{@{}l}{\emph{Sparse graphs}} \\
10 & 10{,}000 & 7{,}791 & 94.51 [94.05, 94.94] & 99.99 [99.93, 100.00] \\
20 & 10{,}000 & 8{,}492 & 93.97 [93.49, 94.42] & 99.94 [99.86, 99.97] \\
40 & 10{,}000 & 9{,}226 & 95.67 [95.25, 96.05] & 99.92 [99.84, 99.96] \\
60 & 10{,}000 & 9{,}518 & 96.67 [96.30, 97.00] & 99.97 [99.91, 99.99] \\
80 & 10{,}000 & 9{,}652 & 97.54 [97.22, 97.83] & 99.97 [99.91, 99.99] \\
160 & 10{,}000 & 9{,}907 & 99.09 [98.88, 99.26] & 100.00 [99.96, 100.00] \\
\addlinespace[4pt]
\multicolumn{5}{@{}l}{\emph{Dense graphs}} \\
10 & 10{,}000 & 8{,}690 & 92.74 [92.21, 93.23] & 99.98 [99.92, 99.99] \\
20 & 10{,}000 & 4{,}627 & 96.56 [96.18, 96.90] & 100.00 [99.92, 100.00] \\
40 & 10{,}000 & 1{,}466 & 99.96 [99.90, 99.98] & 100.00 [99.74, 100.00] \\
60 & 10{,}000 & 858 & 100.00 [99.96, 100.00] & 100.00 [99.55, 100.00] \\
80 & 10{,}000 & 666 & 100.00 [99.96, 100.00] & 100.00 [99.43, 100.00] \\
160 & 10{,}000 & 320 & 100.00 [99.96, 100.00] & 100.00 [98.81, 100.00] \\
\addlinespace[4pt]
\bottomrule
\end{tabular*}
\caption{Coverage of the graphical criteria. Brackets give $95\%$ Wilson confidence intervals.}
\label{tab:criteria-coverage}
\end{table}

\subsection{Sparse and dense graphs}
\label{app:random-graph-design}
We draw random samples of graphs with $n\in\{10,20,40,60,80,160\}$ nodes. For sparse graphs, each directed edge is included with probability $4/n$ and each bidirected edge with probability $2/n$. For dense graphs, these probabilities are $0.30$ and $0.20$. Directed edges follow a fixed topological order, with the last node always serving as the outcome $Y$. The treatment $X$ is then drawn uniformly from the nodes with a directed path to $Y$. For each size and regime, we sample $10{,}000$ graphs in which the target query $p(y|\doo(x))$ is not identifiable.

Table~\ref{tab:criteria-coverage} shows the coverage of the graphical criteria in Section~\ref{sec:graphical-criteria}. Note that Theorems~\ref{thm:bidirected_paths}, \ref{thm:directed_child}, and~\ref{thm:one-bidirected-path} fully resolve $92.74$--$100\%$ of the nonidentified cases. Moreover, Theorems~\ref{thm:bidirected_paths} and~\ref{thm:directed_child} find an ID-enabling edge in at least $99.92\%$ of cases, whenever a solution exists.

We also compare the runtime of Algorithm~\ref{alg:successful-edges} against the naive baseline of expanding and testing every edge. For each case, we record the median runtime across seven timing batches. Speedup is the ratio of the runtime of the brute force procedure to that of Algorithm~\ref{alg:successful-edges}.

The results are shown in Table~\ref{tab:exp-scaling} and Figure~\ref{fig:measurement-runtime}. Across graph sizes, median speedups range from $2.97$ to $5.69$ for sparse graphs and from $2.21$ to $1{,}895.97$ for dense graphs. Figure~\ref{fig:exp-paired} compares the runtimes for individual cases.

\begin{table}[htbp]
\centering\small
\setlength{\tabcolsep}{4pt}
\renewcommand{\arraystretch}{1.12}
\begin{tabular*}{\linewidth}{@{}r@{\extracolsep{\fill}}rrrrrr@{}}
\toprule
 & \multicolumn{3}{c}{Sparse graphs} & \multicolumn{3}{c}{Dense graphs} \\
\cmidrule(lr){2-4}\cmidrule(l){5-7}
 & & \multicolumn{2}{c}{Speedup} & & \multicolumn{2}{c}{Speedup} \\
$n$ & ID-enabling (\%) & All & ID-enabling & ID-enabling (\%) & All & ID-enabling \\
\midrule
10 & 77.91 & 2.97 & 2.65 & 86.90 & 2.21 & 2.05 \\
20 & 84.92 & 3.37 & 3.11 & 46.27 & 12.89 & 8.63 \\
40 & 92.26 & 3.81 & 3.66 & 14.66 & 299.83 & 122.59 \\
60 & 95.18 & 4.10 & 4.00 & 8.58 & 649.16 & 322.04 \\
80 & 96.52 & 4.50 & 4.42 & 6.66 & 758.02 & 467.87 \\
160 & 99.07 & 5.69 & 5.66 & 3.20 & 1{,}895.97 & 1{,}348.12 \\
\bottomrule
\end{tabular*}
\caption{Runtime results for sparse and dense graphs.}
\label{tab:exp-scaling}
\end{table}

\begin{figure}[!htbp]
\centering
\includegraphics[width=\linewidth]{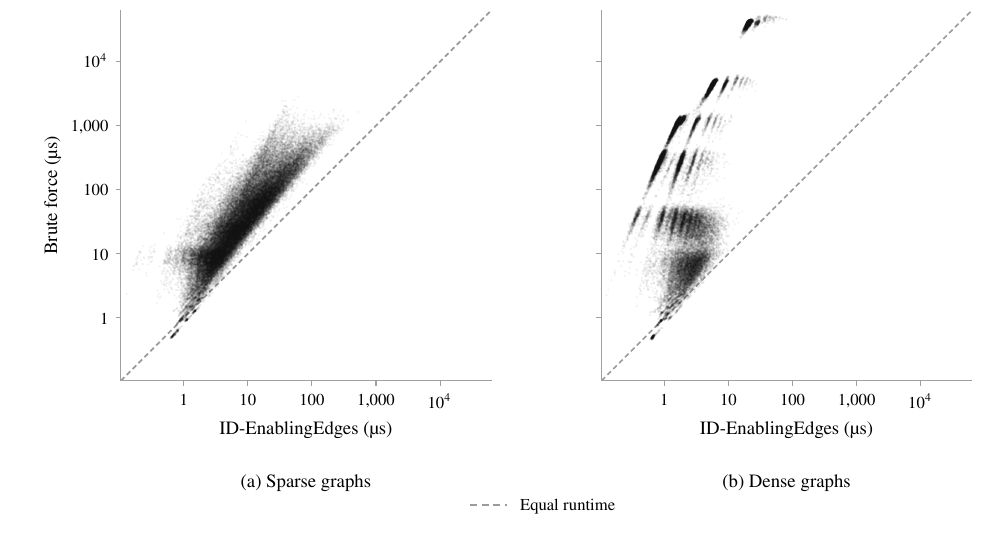}
\caption{Paired runtimes for sparse and dense graphs. Points above the diagonal favor Algorithm~\ref{alg:successful-edges}.}
\label{fig:exp-paired}
\end{figure}

\subsection{Varying edge densities}
We draw random samples of graphs with $n=40$ nodes. We consider all 25 combinations of directed edge probabilities $\{0.05,0.10,0.15,0.20,0.30\}$ and bidirected edge probabilities $\{0.02,0.05,0.10,0.15,0.20\}$, using the same sampling procedure as above. For each combination, we sample $10{,}000$ graphs in which the target query $p(y|\doo(x))$ is not identifiable. Two combinations are already included in the sparse and dense experiments, so we reuse those samples.

Table~\ref{tab:density-coverage} shows the coverage of the graphical criteria in Section~\ref{sec:graphical-criteria}. Note that Theorems~\ref{thm:bidirected_paths}, \ref{thm:directed_child}, and~\ref{thm:one-bidirected-path} fully resolve $87.70$--$99.96\%$ of the nonidentified cases. Moreover, Theorems~\ref{thm:bidirected_paths} and~\ref{thm:directed_child} find an ID-enabling edge in at least $99.92\%$ of cases, whenever a solution exists.

\begin{table}[htbp]
\centering\small
\setlength{\tabcolsep}{4pt}
\renewcommand{\arraystretch}{1.12}
\begin{tabular*}{\linewidth}{@{}r@{\extracolsep{\fill}}rrrr@{}}
\toprule
\shortstack{Directed edge\\probability} & \shortstack{Bidirected edge\\probability} & \shortstack{Queries admitting\\an expansion} & \shortstack{Fully\\resolved (\%)} & \shortstack{Found an\\expansion (\%)} \\
\midrule
0.05 & 0.02 & 9{,}981 & 99.00 [98.79, 99.18] & 100.00 [99.96, 100.00] \\
0.05 & 0.05 & 9{,}919 & 97.39 [97.06, 97.68] & 99.98 [99.93, 99.99] \\
0.05 & 0.10 & 9{,}770 & 92.83 [92.31, 93.32] & 99.98 [99.93, 99.99] \\
0.05 & 0.15 & 9{,}626 & 89.91 [89.30, 90.49] & 99.96 [99.89, 99.98] \\
0.05 & 0.20 & 9{,}509 & 87.70 [87.04, 88.33] & 99.98 [99.92, 99.99] \\
0.10 & 0.02 & 9{,}846 & 99.15 [98.95, 99.31] & 99.98 [99.93, 99.99] \\
0.10 & 0.05 & 9{,}226 & 95.67 [95.25, 96.05] & 99.92 [99.84, 99.96] \\
0.10 & 0.10 & 7{,}796 & 91.43 [90.87, 91.96] & 99.92 [99.83, 99.96] \\
0.10 & 0.15 & 7{,}286 & 91.43 [90.87, 91.96] & 99.93 [99.84, 99.97] \\
0.10 & 0.20 & 7{,}006 & 92.26 [91.72, 92.77] & 99.97 [99.90, 99.99] \\
0.15 & 0.02 & 9{,}582 & 99.43 [99.26, 99.56] & 99.99 [99.94, 100.00] \\
0.15 & 0.05 & 7{,}973 & 96.99 [96.64, 97.31] & 99.94 [99.85, 99.97] \\
0.15 & 0.10 & 5{,}596 & 95.82 [95.41, 96.19] & 99.95 [99.84, 99.98] \\
0.15 & 0.15 & 4{,}625 & 96.67 [96.30, 97.00] & 99.98 [99.88, 100.00] \\
0.15 & 0.20 & 4{,}366 & 97.90 [97.60, 98.16] & 100.00 [99.91, 100.00] \\
0.20 & 0.02 & 9{,}347 & 99.74 [99.62, 99.82] & 100.00 [99.96, 100.00] \\
0.20 & 0.05 & 7{,}102 & 97.64 [97.32, 97.92] & 99.97 [99.90, 99.99] \\
0.20 & 0.10 & 3{,}986 & 97.81 [97.50, 98.08] & 99.97 [99.86, 100.00] \\
0.20 & 0.15 & 3{,}027 & 99.01 [98.80, 99.19] & 99.97 [99.81, 99.99] \\
0.20 & 0.20 & 2{,}805 & 99.27 [99.08, 99.42] & 100.00 [99.86, 100.00] \\
0.30 & 0.02 & 8{,}975 & 99.91 [99.83, 99.95] & 100.00 [99.96, 100.00] \\
0.30 & 0.05 & 5{,}734 & 98.89 [98.67, 99.08] & 99.98 [99.90, 100.00] \\
0.30 & 0.10 & 2{,}524 & 99.50 [99.34, 99.62] & 100.00 [99.85, 100.00] \\
0.30 & 0.15 & 1{,}655 & 99.83 [99.73, 99.89] & 100.00 [99.77, 100.00] \\
0.30 & 0.20 & 1{,}466 & 99.96 [99.90, 99.98] & 100.00 [99.74, 100.00] \\
\bottomrule
\end{tabular*}
\caption{Coverage of the graphical criteria at $n=40$. Brackets give $95\%$ Wilson confidence intervals.}
\label{tab:density-coverage}
\end{table}

We also compare the runtime of Algorithm~\ref{alg:successful-edges} against the naive baseline of expanding and testing every edge, using the same timing procedure as above. The results are shown in Table~\ref{tab:density-runtime}. Across the 25 combinations, median speedups range from $1.74$ to $299.83$.

\begin{table}[htbp]
\centering\small
\setlength{\tabcolsep}{4pt}
\renewcommand{\arraystretch}{1.12}
\begin{tabular*}{\linewidth}{@{}r@{\extracolsep{\fill}}rrrrr@{}}
\toprule
 & \multicolumn{5}{c}{Bidirected edge probability} \\
\cmidrule(l){2-6}
Directed edge probability & 0.02 & 0.05 & 0.10 & 0.15 & 0.20 \\
\midrule
0.05 & 1.74, 1.74 & 1.81, 1.79 & 1.85, 1.82 & 1.88, 1.83 & 1.97, 1.89 \\
0.10 & 4.52, 4.46 & 3.81, 3.66 & 4.80, 4.06 & 7.61, 5.62 & 12.48, 8.14 \\
0.15 & 7.36, 7.08 & 5.75, 5.25 & 10.97, 8.21 & 27.13, 16.53 & 56.84, 29.96 \\
0.20 & 10.40, 9.77 & 7.86, 6.99 & 20.23, 14.71 & 59.63, 35.20 & 124.44, 60.48 \\
0.30 & 16.82, 15.30 & 13.26, 11.31 & 42.25, 32.07 & 139.00, 80.90 & 299.83, 122.59 \\
\bottomrule
\end{tabular*}
\caption{Median paired speedups at $n=40$, with $10{,}000$ nonidentified queries per density combination. Each entry reports the median over all queries, followed by the median among queries admitting an ID-enabling expansion.}
\label{tab:density-runtime}
\end{table}

\subsection{Constructed dense graphs}
\label{app:dense-family}
For each $n\in\{10,20,40,60,80,160\}$, we construct a graph with nodes $Z_1,\ldots,Z_{n-2},X,Y$. We include every directed edge from an earlier to a later node in the order $Z_1,\ldots,Z_{n-2},X$. We also include the edge $X\rightarrow Y$ and the bidirected path $X\leftrightarrow Z_1\leftrightarrow\cdots\leftrightarrow Z_{n-2}\leftrightarrow Y$. The target query $p(y|\doo(x))$ is not identifiable in these graphs.

We compare Algorithm~\ref{alg:successful-edges} against expanding and testing every edge. The results are shown in Figure~\ref{fig:dense-call-growth}. Algorithm~\ref{alg:successful-edges} performs $n+1$ final ID tests, whereas the brute force procedure performs $\binom{n-1}{2}+n$. At $n=160$, for example, these counts are $161$ and $12{,}721$, respectively.

\begin{figure}[t]
\centering
\includegraphics[width=0.6\linewidth]{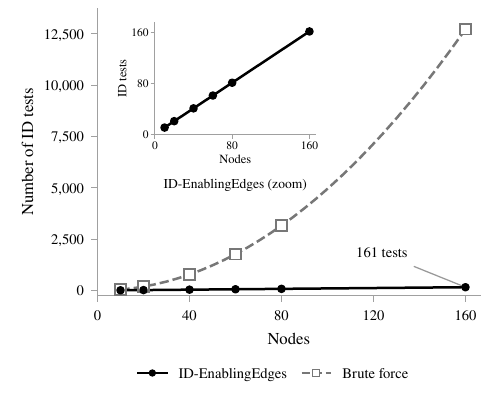}
\caption{Number of final ID tests in the constructed dense graphs.}
\label{fig:dense-call-growth}
\end{figure}

\subsection{Computing environment and replication}
All experiments were run on an Apple M5 Max computer with 18 CPU cores and 64 GB of memory, running macOS 27.0. The algorithms are implemented in C++17 and compiled with Apple clang 21.0 at optimization level \texttt{-O3}. The experiments are run and analyzed in Python 3.12 using NumPy, PyArrow, and Matplotlib. The exhaustive and random experiments were evaluated with 12 parallel workers. Runtimes were then measured on a single thread, with no other experiments running. The complete run took approximately two hours, including correctness checks and runtime measurements, and produced 5.7 GB of saved records. We provide code that reproduces all experiments, tables, and figures from scratch, using a fixed random seed {as Online Supplementary Material.}
\end{toappendix}

\begin{toappendix}
\section{Variants of Algorithm~\ref{alg:successful-edges}}
\label{app:variants}

Here we generalize Algorithm~\ref{alg:successful-edges} by allowing other constructions of the candidate set $\+ C$.

\begin{algorithm}[htbp]
\caption{General form of Algorithm~\ref{alg:successful-edges}.}
\label{alg:general}
\begin{algorithmic}[1]
\Function{ID-EnablingEdges}{$\+ Y,\+ X,\cG$}
    \State $\+ C \gets$ a set of edges of $\cG$ containing every ID-enabling edge
    \State $\+ S \gets \emptyset$
    \For{$e \in \+ C$}
        \If{$\Call{ID}{\+ Y,\+ X,\cG_e}$ does not fail}
            \State $\+ S \gets \+ S \cup \{e\}$
        \EndIf
    \EndFor
    \State \Return $\+ S$
\EndFunction
\end{algorithmic}
\end{algorithm}

We compare variants of Algorithm~\ref{alg:successful-edges} using four necessary conditions for an edge $e$ to be ID-enabling:
\begin{enumerate}[label=(\roman*)]
    \item $e$ belongs to $\cG[\An(\+ Y)_{\cG}]$.
    \item $e$ belongs to every larger graph $\cC_i$ in the pairs returned by HCID.
    \item $e$ is weak directed or weak bidirected in every corresponding ID-hedge.
    \item If $\pi$ is a bidirected path from $X_i\in\+ X$ to a child $C$ of $X_i$, and all nodes of $\pi$ other than $X_i$ lie in $\An(\+ Y)_{\cG_{\bar{\+ X}}}\setminus\+ X$, then
    \[
        e\in E(\pi)\cup\{(X_i,C)\}.
    \]
\end{enumerate}
These conditions follow from Corollary~\ref{thm:characterization} and Theorems~\ref{thm:successful-breaks}, \ref{thm:weak-edges}, and~\ref{thm:one-bidirected-path}, respectively.

The brute force procedure uses (i) alone, while Algorithm~\ref{alg:successful-edges} uses (i)--(iii). Table~\ref{tab:variants} compares the mean number of final ID tests for these and other combinations, including the addition of (iv).

Table~\ref{tab:variants} shows that Theorem~\ref{thm:one-bidirected-path} can be a powerful filter to complement Algorithm~\ref{alg:successful-edges}, but on its own it does not prevent a quadratic number of final ID tests in the worst case. The weak edge conditions used by Algorithm~\ref{alg:successful-edges} guarantee at most $2n-2$ such tests.

\begin{table}[htbp]
\centering\small
\setlength{\tabcolsep}{6pt}
\begin{tabular}{rrrrrrr}
\toprule
 & \multicolumn{5}{c}{Conditions used} & \\
\cmidrule(lr){2-6}
$n$ & (i) & (i)--(ii) & (i)--(iii) & (i)--(iv) & (i) + (iv) & $|\+ S|$ \\
\midrule
\multicolumn{7}{l}{\emph{Exhaustive experiments}} \\
4 & 7.65 & 6.60 & 2.74 & 1.51 & 1.67 & 1.48 \\
5 & 11.86 & 10.55 & 2.70 & 1.17 & 1.42 & 1.11 \\
6 & 17.07 & 15.83 & 2.38 & 0.85 & 1.13 & 0.79 \\
\addlinespace
\multicolumn{7}{l}{\emph{Sparse graphs}} \\
10 & 20.35 & 13.47 & 4.14 & 1.72 & 2.08 & 1.64 \\
20 & 32.59 & 17.31 & 6.01 & 2.29 & 3.04 & 2.12 \\
40 & 47.55 & 17.86 & 7.56 & 2.81 & 3.68 & 2.66 \\
60 & 57.16 & 16.78 & 8.08 & 3.04 & 3.99 & 2.90 \\
80 & 62.93 & 14.18 & 7.75 & 3.07 & 3.85 & 2.97 \\
160 & 71.90 & 8.14 & 5.85 & 2.87 & 3.27 & 2.84 \\
\addlinespace
\multicolumn{7}{l}{\emph{Dense graphs}} \\
10 & 14.76 & 10.00 & 4.06 & 2.02 & 2.38 & 1.91 \\
20 & 63.53 & 58.60 & 3.67 & 0.79 & 1.28 & 0.71 \\
40 & 308.15 & 307.70 & 0.48 & 0.15 & 0.15 & 0.15 \\
60 & 755.09 & 755.08 & 0.16 & 0.09 & 0.09 & 0.09 \\
80 & 1{,}403.40 & 1{,}403.40 & 0.11 & 0.07 & 0.07 & 0.07 \\
160 & 5{,}988.83 & 5{,}988.83 & 0.06 & 0.03 & 0.03 & 0.03 \\
\addlinespace
\multicolumn{7}{l}{\emph{Constructed dense family}} \\
10 & 46 & 46 & 11 & 11 & 46 & 11 \\
20 & 191 & 191 & 21 & 21 & 191 & 21 \\
40 & 781 & 781 & 41 & 41 & 781 & 41 \\
60 & 1{,}771 & 1{,}771 & 61 & 61 & 1{,}771 & 61 \\
80 & 3{,}161 & 3{,}161 & 81 & 81 & 3{,}161 & 81 \\
160 & 12{,}721 & 12{,}721 & 161 & 161 & 12{,}721 & 161 \\
\bottomrule
\end{tabular}
\caption{Mean number of final ID tests per nonidentified query using different combinations of necessary conditions. The last column gives the mean number of ID-enabling edges. Constructed-family rows give exact counts for the graphs in Appendix~\ref{app:dense-family}. Case~(i) matches the brute force approach of testing every edge. Algorithm~\ref{alg:successful-edges} uses (i)--(iii) (third column).}
\label{tab:variants}
\end{table}
\end{toappendix}
\section{Conclusion}

\begin{figure}[t]
\centering
\includegraphics[width=\linewidth]{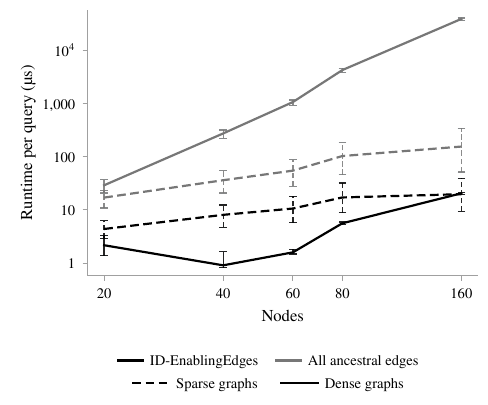}
\caption{Runtime on nonidentified queries. Curves show medians and vertical bars show interquartile ranges over $10{,}000$ queries per setting.
}
\label{fig:measurement-runtime}
\end{figure}

We introduced a framework for determining where further investigation of a causal mechanism or source of confounding could make a previously nonidentifiable effect identifiable in a semi-Markovian causal model. We gave simple graphical conditions for finding ID-enabling edges and ruling out edges whose expansion cannot enable identification. By relating failures of the ID algorithm to ID-hedges, we characterized the directed and bidirected edges that can be expanded to eliminate the obstructions to identification. This characterization leads to a sound and complete algorithm that finds all ID-enabling edges. The algorithm runs final identification tests on at most \(2n-2\) candidate edges, and this bound is sharp. In our experiments, the graphical conditions fully resolve most nonidentified queries, and the algorithm can be more than a thousand times faster than testing every edge. Identification theory is usually used to decide whether a given causal model identifies the effect of interest. Our results show that it can also guide how the model should be refined, by indicating which mechanisms or sources of confounding, once measured, would make the effect identifiable. Extending the approach to several simultaneous edge expansions is left for future work.

\newpage

\subsection*{AI use statement}

We used generative AI tools to copyedit the paper, conduct extensive adversarial checks of our arguments and results, and write performant code for the simulations.
We take responsibility for the final content of this work, including text, claims, or artifacts produced with the aid of generative AI.

\bibliography{bibliography}

\clearpage
\appendix
\thispagestyle{empty}

\onecolumn
\aistatstitle{
Supplementary Materials}

\end{document}